\documentclass[draftcls, journal, onecolumn, 12 pt]{IEEEtran}
\usepackage{amsfonts, cite}
\usepackage{amssymb}
\usepackage{color}
\usepackage{subfigure}
\usepackage{scalefnt}
\usepackage{amsmath}
\usepackage{amsthm} 
\usepackage{amsfonts}
\usepackage{colordvi}
\usepackage{verbatim}
\usepackage{graphicx}
\usepackage{latexsym}
\usepackage{bm}
\usepackage{cite}

\usepackage{times}

\usepackage{algpseudocode}

\newcommand{\Ex}{\mathbb{E}\/}
\newcommand{\Px}{\mathbb{P}\/}
\newcommand{\CN}{\mathcal{CN}}
\newcommand{\LN}{\mathcal{LN}}

\usepackage{paralist}
\defaultleftmargin{1em}{}{}{}

\newtheorem*{theorem*}{Theorem}

\newtheorem*{proposition*}{Proposition}

\title{Analyzing the Impact of Access Point Density on the Performance of Finite-Area Networks}%

\author{S. Alireza~Banani, Andrew W.~Eckford, Raviraj~S.~Adve%

\thanks {This article has been accepted for publication in a future issue of this journal, but has not been fully edited. Content may change prior to final publication. Citation information: DOI 10.1109/TCOMM.2015.2481887, IEEE Transactions on Communications.}
\thanks{This research was supported by TELUS and the Natural Sciences and Engineering Research Council (NSERC).}
\thanks{S. A. Banani and R. S. Adve are with the ~Dept. of Electrical and Computer Engineering, University of Toronto, Toronto, Canada; emails: alireza.banani@utoronto.ca, rsadve@comm.utoronto.ca}%
\thanks{A. Eckford is with the Dept. of Electrical Engineering and Computer Science, York University, Toronto, Canada; email: aeckford@yorku.ca}}

\IEEEaftertitletext{\vspace{-1.25\baselineskip}}

\begin{document}
\bibliographystyle{ieeetr}
\maketitle \thispagestyle{empty}

\begin{abstract}
Assuming a network of infinite extent, several researchers have analyzed small-cell networks using a Poisson point process (PPP) location model, leading to simple analytic expressions. The general assumption has been that these results apply to finite-area networks as well. However, do the results of infinite-area networks apply to finite-area networks? In this paper, we answer this question by obtaining an accurate approximation for the achievable signal-to-interference-plus-noise ratio (SINR) and user capacity in the downlink of a \textit{finite-area} network with \textit{a fixed number of} access points (APs). The APs are uniformly distributed within the area of interest. Our analysis shows that, crucially, the results of infinite-area networks are very different from those for finite-area networks of low-to-medium AP density. Comprehensive simulations are used to illustrate the accuracy of our analysis. For practical values of signal transmit powers and AP densities, the analytic expressions capture the behavior of the system well. As an added benefit, the formulations developed here can be used in parametric studies for network design. Here, the analysis is used to obtain the required number of APs to guarantee a desired target capacity in a finite-area network.

\begin{IEEEkeywords}
Finite-area networks, downlink, coverage probability, small cells, Moment Matching Approximation, Poisson point process
\end{IEEEkeywords}

\end{abstract}

\section{Introduction}
\label{sec:introduction}
As the available user capacity in traditional cellular systems has saturated, the wireless industry is planning on the introduction of small-cell networks, including outdoor access points (APs) and/or indoor femtocell APs. With limited location planning possible in such networks, these APs are placed in an irregular manner; the APs are modeled as having random locations. Importantly, the available analysis techniques largely focus on the asymptotic case of networks of infinite extent. Our motivation, on the other hand, is analyzing finite area networks such as networks that provide coverage inside buildings, or at outdoor hotspots. Given the lack of accurate and tractable analysis techniques for \emph{finite-area} networks with a finite number of APs, it has generally been assumed that the infinite-network results directly apply~\cite{mao2012towards}. However, as our work will show, for practical values of system parameters, this is not always true. We will analyze this discrepancy in the context of metrics relevant to a network designer.

\subsection{Literature Survey and Motivation}
Traditional network models are either impractically simple (e.g., the Wyner model~\cite{wyner1994shannon}) or excessively complex (e.g., general case of random user location with APs on a hexagonal lattice~\cite{goldsmith2005wireless}) to accurately model small-cell networks. A useful mathematical model that accounts for the randomness in AP locations (and irregularity in the cells) uses spatial point processes, such as the Poisson point process (PPP), to model the location of APs in the network~\cite{andrews2011tractable,huang2013analytical, dhillon2012modeling,dhillon2013downlink,di2013average,heath2013modeling}. This allows for the use of techniques from stochastic geometry~\cite{haenggi2009stochastic,win2009mathematical,chiu2013stochastic} and large-deviation theory~\cite{huang2013analytical} to characterize the distribution of the signal-to-interference-plus-noise-ratio (SINR) and/or user capacity in large networks. For example, assuming the networks are of infinite extent, rate expressions are available, e.g.~in~\cite{dhillon2013downlinkrate,madhusudhanan2012downlink, blaszczyszyn2013equivalence}, while accounting for path loss, small-scale fading and log-normal shadowing.

A review of different network models in the literature is helpful in understanding the motivation for our work. Two models are relevant here: the infinite-network model, as the name suggests, assumes a network of infinite geographical extent usually with a fixed AP density;
on the other hand, the dense-network model
%
considers a finite area with large AP density. Both have been widely used in the asymptotic analysis of networks (asymptotic in the number of APs)~\cite{mao2012towards}. The assumption of an infinitely large network, coupled with AP locations modeled as a PPP, allows for analytic tractability.

Although such infinite-area analyses provide convenient closed-form expressions, they do not completely reflect the more realistic case of a \textit{finite-area} network, especially with a low AP density or \textit{finite number} of APs. As recent work has shown, treating a finite-area network as spread over an infinite area is accurate for cases with very high pathloss exponents (e.g., $\alpha =6$); for more realistic values such as the range of $2 \leq \alpha \leq 4$, the infinite-area assumption underestimates network performance significantly~\cite{vijayandran2012analysis}.

\begin{figure}[t]
 \begin{center}
\includegraphics[width = .5\textwidth]{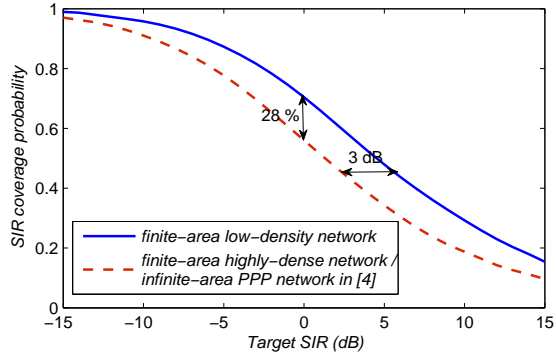}
 \caption{Comparison of the SIR coverage probability obtained at the centre of a circular finite-area interference-limited network with AP densities of $\lambda = 1~\textmd{AP/km}^2$ and $\lambda = 30~\textmd{APs/km}^2$ under Rayleigh fading and no shadowing with PLE of $\alpha =3.87$.}
 \label{Fig_1}
 \end{center}
 \end{figure}

To motivate this paper, Fig.~\ref{Fig_1} illustrates this issue via simulations. For a circular interference-limited network of radius $1~\textmd{km}$, the figure plots the signal-to-interference ratio (SIR) coverage probability at the network center with AP densities of $\lambda = 1~\textmd{AP/km}^2$ and $\lambda = 30~\textmd{APs/km}^2$ for the path loss exponent (PLE) of $\alpha =3.87$ under purely Rayleigh fading. Note that in interference-limited infinite-area networks, coverage probability (CP) is not a function of density, i.e., the ``infinite-area PPP" curve, obtained from the work in~\cite{andrews2011tractable}, is valid for all AP densities. The PPP curve, for an infinite-area network, exactly matches the curve for the finite-area high-density network ($\lambda = 30~\textmd{APs/km}^2$), i.e., a highly dense finite-area network can be closely approximated as a network of infinite extent. However, as seen from the figure, the finite-area low-density network outperforms the dense network by 28\% in SIR CP for the target SIR of 0 dB (or by 3~dB in SIR).

There are only a few works that have investigated the performance of finite-area networks~\cite{vijayandran2012analysis,srinivasa2007modeling, salbaroli2009interference,torrieri2012outage,torrieri2013analysis}. For example, in~\cite{vijayandran2012analysis,srinivasa2007modeling}, the authors use a moment generating function (MGF) approach to characterize the interference under Rayleigh fading. Another set of related works with APs distributed as a PPP, can be found in~\cite{torrieri2012outage, torrieri2013analysis} where closed-form expressions are obtained for the instantaneous outage probability for a \textit{given realization} of AP locations. The authors then use Monte Carlo simulations to obtain the outage probability averaged over network realizations.

\subsection{Our Contributions}
Our goal in this paper is to analyze the performance of a circular\footnote{The choice of a circular area is for simplicity that can lead to tractable analysis; given relevant distance probability distribution functions in the literature, other geometric shapes could be analyzed.} finite-area network with a finite number of APs. Our motivating examples are indoor networks~\cite{nitaigour2007sensor, quek2013small} and outdoor hotspots~\cite{bazzi2006wlan, liang2011dfmac} . We wish to provide a network designer the ability to quickly analyze the impact of various network parameters. Specifically, we develop an expression which closely approximates the SINR and user capacity at any point in a circular, finite-area, network serviced by a fixed number of APs. Our work differs from that available in the literature (for finite areas) in fixing the number of APs, allowing for a random connection distance and accounting for shadowing and noise. This difference adds complexity to the analysis, but better represents the problem at hand.
In our model, the $N$ APs within the circle are uniformly distributed\footnote{The uniform distribution is equivalent to a homogenous PPP (of corresponding density) \emph{conditioned on having $N$ APs within the circle} (in this case the conditional PPP is a binomial point process~\cite{moller2004statistical, illian2008statistical}).}. As in~\cite{andrews2011tractable}, we analyze downlink transmissions where independent users are associated with their \textit{nearest} AP, while all other APs act as interferers. We obtain the coverage probability within any point inside the circle. Unlike previous works on finite-area networks that focus on Rayleigh fading exclusively, our model accounts for path-loss, small-scale fading and shadowing.

To confirm the accuracy of our analysis, we compare our analytic results with that of Monte-Carlo simulations. For practical values of transmit signal powers and AP densities, our approximations capture the behaviour of the network very well. Our results match the reports presented earlier~\cite{vijayandran2012analysis}, in that, in the interference-limited case, the SIR coverage probability performance of an infinite-area network (equivalently, dense network) underestimates that of a low-density network. As an added benefit, the expressions developed here allow us to quantify the gains, in terms of coverage probability, by adding APs within the circle.

Motivated by the desire to provide network design tools, we then focus on the origin of the circular area - this location has the worst user capacity. The worst-case user capacity has been used, e.g. in~\cite{goldsmith2005wireless, rappaport1996wireless, karakayali2006network,banani2013analyzing}, for network design in wireless networks with and without cooperation amongst APs. For the special choice of PLE $\alpha=4$, we derive a closed-form expression for the worst-case user capacity. As a design example, the user capacity at the worst-case point is used to obtain the number of APs required to guarantee a minimum coverage probability everywhere within the area under consideration. This corresponds to designing a finite-area network with a coverage guarantee.

We note that this paper differs from the works in~\cite{vijayandran2012analysis,srinivasa2007modeling} in three aspects:
\begin{itemize}
 \item We allow for a random connection distance from user to its serving AP which better accounts for the randomness in AP locations and irregularity in the cells.
 \item In~\cite{vijayandran2012analysis,srinivasa2007modeling} the number of APs falling in the chosen area is random whereas in the proposed work is fixed.
 \item While in~\cite{vijayandran2012analysis,srinivasa2007modeling} the authors use the MGF approach to characterize the interference under Rayleigh fading, here we use the moment matching approximation~\cite{pratesi2006generalized} which allows us to account for shadowing as well as small-scale fading.
\end{itemize}

\subsection{Organization and Notation}
The rest of the paper is organized as follows: Section~\ref{sec:SysModel} describes the system under consideration. The analysis, the main contribution of this paper, is presented in Section~\ref{sec:SINR}; while supporting simulation results are presented in Section~\ref{sec:Sims}. Finally, Section~\ref{sec:Conc} summarizes and concludes the paper.
The notation used is conventional: matrices are represented using bold upper case and vectors using bold lower case letters; $(\cdot)^{H}$ , and $(\cdot)^{T}$ denote the conjugate transpose, and transpose, respectively. $a \sim \CN(\mu,\sigma^2)$ or $\sim \mathcal{N}(\mu,\sigma^2)$ denote complex and real Gaussian random variables, respectively, with mean $\mu$ and variance $\sigma^2$ while $X\sim\LN(\mu_{x},\sigma_x^2)$ represents a log-normal random where $\ln(X) \sim \mathcal{N}(\mu_{x},\sigma_x^2)$. $Q(x)$ represents the standard \emph{Q}-function, the area under the tail of the standard normal distribution i.e., $\mathcal{N}(0,1)$; $f(\cdot)$ denotes a probability density function (PDF) while $F(\cdot)$ denotes the cumulative distribution function (CDF). Finally, $\Px\{\cdot \, \}$ denotes the probability of an event and $\Ex\{\cdot \,\}$ denotes expectation.

\section{Downlink System Model} \label{sec:SysModel}

\subsection{Assumptions and Initial Analysis}

In this paper we develop an analytical formulation of achievable SINR and user capacity within a finite-area network for a given AP density. The analysis is based on some simplifying approximations and assumptions which are summarized here:

\begin{itemize}
 \item We focus on the downlink of a single-tier finite-area reuse-1 network comprising $N$ APs located in a circular area $\textmd{W}$ with radius $R_{\textmd{W}}$.
 \item The APs are uniformly distributed within the circle; this partitions the circular area into Voronoi cells.
 \item All APs transmit at a power level of $\sigma_s^2$.
 \item Users are associated with the closest AP, i.e., users within a specific Voronoi cell connect to the AP within that cell\footnote{The choice of user association based on minimum distance criterion is for simplicity that can lead to tractable analysis; an improvement in performance can be obtained with the user association based on the strongest received power in a network under shadowing~\cite{dhillon2013downlinkrate,singh2014joint}}.
 \item There are a large number of users uniformly distributed within the network and so each AP is fully loaded serving an equal number of users at a given time. If $K$ denotes the total number of users to be served in the network, each AP serves $K/N$ users at any given time over $K/N$ frequency slots (one frequency slot per user).
 \item For a given $N$, the bandwidth per user served is fixed, i.e., the total bandwidth is divided into $K/N$ equal frequency slots. The total bandwidth is $K\bar{W}$, and each user is allocated a bandwidth of $W_0 = K\bar{W}/(K/N) = N\bar{W}$\footnote{Essentially, the assumption here is that as the number of APs increases, each AP serves fewer users and the bandwidth per user is linearly proportional to the number of APs.}.

    \emph{Comment}: Essentially we are assuming that since there are a large number of users uniformly distributed within the network, at any given time, there exist at least $K/N$ users within an specific Voronoi cell connecting to the AP in that cell. Furthermore, since the average number of users per AP is inversely proportional to the number of APs ($N$), we assume that the bandwidth available \emph{per user} is proportional to $N$.

    We note that other reasonable loading models are possible~\cite{singh2014joint,singh2013offloading}. For example, we could make the bandwidth per user, $W_0$, independent of $N$. Alternatively, one could make the number of users per AP a random variable and each user allocated a random bandwidth; this is very hard to analyze. One could also consider a fixed bandwidth per user with some bandwidth ``wasted" at APs with few associated users. Such a scenario could be analyzed within our framework using a thinned PPP. Our choice is based on the intuition that, as $N$ increases, fewer users are served by each AP, and so, more bandwidth should be available per user.

\end{itemize}

\begin{figure}[t]
\begin{center}
\includegraphics[width=0.49\textwidth]{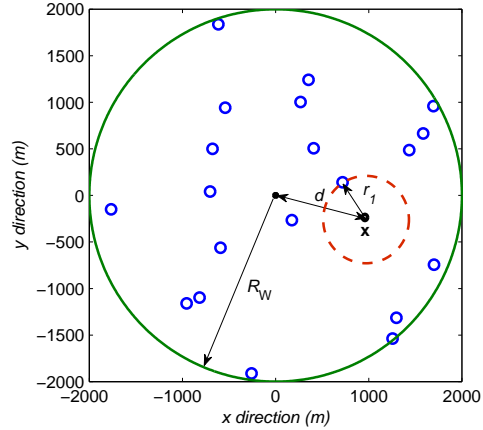}
\caption{One realization of the location of APs based on a PPP in a circular region with radius $R_{\textmd{W}} =2~ \textmd{km}$ and $N=20$. In the figure, $r_{1}$ and $d$ are the distances of an arbitrary point $\textbf{x}$ to the closest AP and to the centre of the circular area, respectively.}
\label{Fig_2}
\end{center}
\end{figure}

Figure~\ref{Fig_2} depicts one realization of our network with $N=20$ and $R_{\textmd{W}} =2~\textmd{km}$. Let $h_{1} $ denote the instantaneous channel from a user located at an arbitrary point $\textbf{x}$ to the nearest AP a distance $r_{1}$ away. Further, let $h_{j}, j=2,\cdots,N$ denote the corresponding channels between the user and the remaining $N-1$ interfering APs. Note that these APs are located outside the circle centered at $\textbf{x}$ with radius $r_{1}$ (between the circle centered at $\textbf{x}$ with radius $r_{1} $ and circle with radius $R_{\textmd{W}}$). Similarly, let $r_{j}, j=2,\cdots,N$ represent the distance\footnote{We note that $r_{j}, j=2,\cdots,N$ are un-sorted, independent and identically distributed, distances to the $N-1$ interfering APs.} from the \textit{j-}th AP to the user at point $\textbf{x}$, and let $\textmd{PL}(r_{j})$ represent the path loss (in dB) over this distance. The instantaneous channels $h_{j} , j=1,\cdots,N$ are modeled as
\begin{equation} \label{eq_2}
h_{j} =\bar{h}_{j} \times 10^{-(\textmd{PL}(r_{j} )+L_{j} )/20},
\end{equation}
where $\bar{h}_{j} \sim \CN(0,1)$ represents the normalized complex channel gain, reflecting small-scale Rayleigh fading, from the  \textit{j}-th AP to the user, which is independent from $\bar{h}_{i},i\ne j$; and where $L_{j} {\kern 1pt} \sim \mathcal{N}(0,\sigma _{L} )$ models the large-scale fading or shadowing, modeled as a lognormal random variable. The standard deviation (STD) $\sigma _{L}^{}$ is expressed in dB. The path loss, in dB, is given by $\textmd{PL}(r_{j})=10\,\alpha \log _{10}r_{j}$, where $\alpha$ denotes the path loss exponent.

Let $\sigma _{n}^{2}$ represent the power of the thermal noise, and $I_{r_{1}} $ denote the power of the interference from the $N-1$ interfering APs. The instantaneous SINR of the user at a random distance $r_{1}$ from its nearest AP can be expressed as
\begin{equation} \label{eq_3}
\mathtt{SINR}_{\,r_{1}} =\frac{\sigma _{s}^{2} \left|h_{{\kern 1pt} 1} \right|^{2} }{\sigma _{n}^{2} + I_{r_{1}}} =\frac{\sigma _{s}^{2} \left|\bar{h}_{{\kern 1pt} 1} \right|^{2} r_{1}^{-\alpha } z_{1} }{\sigma _{n}^{2} + \sum _{j=2}^{N}\sigma _{s}^{2} \left|\bar{h}_{j} \right|^{2} r_{j}^{-\alpha } z_{j}},
\end{equation}
where $z_{j} =10^{-L_{j} /10}, j=1,\cdots,N$ are independent lognormal RVs as $z_{j} \sim \LN(\mu_{z}=0, \sigma _{z} =(0.1\ln 10) \sigma _{L})$. Thus, the instantaneous achieved SINR depends on $r_{1}$ (both via $\textmd{PL}(r_{1})$  and $I_{r_{1}} $) as well as the instantaneous realizations of $\bar{h}_{j}$, and $L_{j}, j=1,\cdots,N$. It is known that in an \textit{infinite} area with \textit{infinite} number of PPP distributed APs, the interference follows an alpha stable distribution~\cite{andrews2011tractable,haenggi2009interference}. However, this is not true for a \emph{finite-area with a finite number of APs}; this necessitates a new analysis technique. Here we present an accurate analysis.

With a fixed bandwidth of $W_{0} = N\bar{W}~\textrm{Hz}$ available to each user, an instantaneous per-user data rate (in b/s) of $R_{\,r_{1}} = W_{0}\log_{2}(1+\mathtt{SINR_{\,r_{1}}})$ is achievable. The rate coverage probability, defined as the probability that the user can achieve a target rate $R_0$, is given by
\begin{equation} \label{eq_4}
\begin{split}
\Px\{R_{\,r_{1}} > R_{0}\} = & \Px\left\{N\bar{W}\log_{2}\left(1+\mathtt{SINR_{\,r_{1}}}\right) >
R_{0}\right\} \\
= & \Px\left\{N \log_{2}\left(1+\mathtt{SINR_{\,r_{1}}}\right) >
R_{0}/\bar{W} \right\} \\
= & \Px\{C_{\,r_{1}} > C_{0}\} = \Px\left\{\mathtt{SINR_{\,r_{1}}} > 2^{C_{0}/N}-1)\right\},
\end{split}
\end{equation}
where we define $C_{\,r_{1}} = N\log_{2}\left(1+\mathtt{SINR_{\,r_{1}}}\right)$ and $C_0 = R_0/\bar{W}$, as the achievable and required spectral efficiencies (in b/s/Hz).

\subsection{User Distance Distributions}

To characterize the signal component of the SINR, we need to obtain the user distance to the nearest AP. The first step is to obtain the unconditional distance CDF from point $\textbf{x}$ to an arbitrary AP randomly placed in the circular finite-area. Let $d \leq R_{\textmd{W}}$ denote the distance of $\textbf{x}$ to the centre of the circular area. The CDF of the distance between point $\textbf{x}$ to an arbitrary AP randomly located in the circular region, independently from other $N-1$ APs, is given by~\cite{khalid2012distance,bettstetter2004connectivity},
\begin{equation} \label{eq_9_1}
\begin{split}
F_{R\,}(r) \! = \Px\{R \leq r\} = \! \left\{
\begin{array}{l l c}
\!\! r^2 / R^{\,2}_{\textmd{W}} & ; & 0 \leq r \leq R_{\textmd{W}}-d \\
\!\! \bar{F}_{R}(r) & ; & \! R_{\textmd{W}} - d  \leq r \leq R_{\textmd{W}} + d \\
\! \! 1 & ; & \!  R_{\textmd{W}} +d \leq r
\end{array}
\right.
\end{split}
\end{equation}
where $\bar{F}_{R}(r)$ is given by
\begin{equation} \label{eq_7}
\begin{split}
&\bar{F}_{R\,}\!(r) \!= \!\frac{1}{\pi}\cos^{-1}\left(\frac{d^{2} \! -\! r^{2}\!+\!R^{\,2}_{\textmd{W}}}{2d\,R_{\textmd{W}}}\right)  +\frac{r^{2}}{\pi R^{\,2}_{\textmd{W}}}
\cos^{-1}\left(\frac{d^{2}\!+\!r^{2}\!-\!R^{\,2}_{\textmd{W}}}{2d\,r}\right) \\
& \hspace{1.5in} -\frac{1}{2\pi R^{\,2}_{\textmd{W}}}
\sqrt{((R_{\textmd{W}}+r)^{2}-d^{2})(d^{2}-(r-R_{\textmd{W}})^{2})}.
\end{split}
\end{equation}

Now that any of the $N$ independent APs has a distance CDF as in~\eqref{eq_9_1}, the CDF of the minimum distance - corresponding to the distance from point $\textbf{x}$ to the closest AP - is given by
\begin{equation} \label{eq_10}
F_{\,R_{\,1}}\!(r_{1}) = 1 - [1- F_{R\,}(\,r_{1})]^N,
\end{equation}
where $r_{\,1}$ is the distance of $\textbf{x}$ to the nearest AP.

It is worth noting that, unlike $r_{1}$, the distances $r_{j},j=2,\cdots,N$ are i.i.d.~RVs, but with distance CDFs that differ from $F_{R_{1}}(r_{1} )$. For a given $r_{1}$, the $N-1$ interfering APs are located in the area between circles with radii $r_{1}$ centred at point $\textbf{x}$ and circle with radius $R_{\textmd{W}}$. Therefore, the conditional CDF of $r_{j}$, given $r_{1}$, is
\begin{equation} \label{eq_11}
\begin{split}
F_{R_{j}|\,r_{1}}(r_{j}) = & \Px\{R \leq r_{j} | R > r_{1}\}=\frac{\Px\{R \leq r_{j} \bigcap R > r_{1}\}}{\Px\{R>r_{1}\}}\\
= &  \left\{
\begin{array}{l l c}
\!\! 0 & ; & r_{j} \leq r_{1}  \\
\!\! (F_{R}(r_{j})-F_{R}(r_{1}))/(1-F_{R}(r_{1})) & ; & \!  r_{1} < r_{j} \leq R_{\textmd{W}} + d \\
\!\! 1 & ; & \!  R_{\textmd{W}} + d < r_{j} \\
\end{array}
\right.
\end{split}
\end{equation}
where $F_{R\,}(\,\cdot \,)$ is given in~\eqref{eq_9_1}.

\section{SINR and User Capacity}\label{sec:SINR}

Using the results in the previous section, we now obtain the user capacity in an interference-limited network, i.e., we first assume that the thermal noise is negligible as compared to the interference and can be hence ignored. While this may be justified in dense small cell networks~\cite{boudreau2009interference}, we then generalize the formulation to include thermal noise.

\subsection{Interference-limited network: $\sigma^2_n = 0$}

\begin{proposition*} \label{thm:DistanceCDF}
The averaged SIR coverage probability (averaged over different realizations of AP locations) at an arbitrary point $\textbf{x}$ within the finite-area network is accurately approximated as
\begin{equation} \label{eq_24}
\textmd{CP}_{\mathtt{SIR}}^{\textmd{\,avg}} (N, d, T, \alpha, \sigma_L) =  \Px\{ \mathtt{SIR} > T \}
\simeq \int_{0}^{R_{\textmd{W}} + d } Q\left( \frac{\ln T - \mu _{\mathtt{SIR}}}{\sigma _{\mathtt{SIR}}} \right) \frac{d F_{R_{1}}(r_{1})}{d\,r_{1}} \,d\,r_{1}.
\end{equation}
where $T$ is a chosen SIR threshold and
\begin{equation} \label{eq_mu_SIR}
\mu _{\mathtt{SIR}} = \ln (\sigma _{s}^{2} /\sqrt{2} )+\ln r_{1} ^{-\alpha} - 2\ln (M_{1} ) + 0.5\ln (M_{2}),
\end{equation}
\begin{equation} \label{eq_sigma_SIR}
\sigma _{\mathtt{SIR}}^{2} = \ln 2+\sigma _{z}^{2} - 2\ln (M_{1} )+\ln (M_{2}),
\end{equation}
with
\begin{equation} \label{eq_21}
M_{1} = (N-1)\sigma _{s}^{2} e^{\sigma _{z}^{2} /2} \left[ (R_{\textmd{W}}+d)^{-\alpha} +  \int_{(R_{\textmd{W}} + d)^{-\alpha} }^{ r_{1}^{-\alpha}} \frac{F_{R}(s_{j}^{-1/\alpha})-F_{R}(r_{1})}{1-F_{R}(r_{1})} \,\textmd{d} s_{j} \right],
\end{equation}
\begin{equation} \label{eq_22}
\begin{split}
M_{2} = & 2(N-1)\sigma _{s}^{4} e^{2\sigma _{z}^{2}} \left[ (R_{\textmd{W}}+d)^{-2\alpha} +  \int_{(R_{\textmd{W}} + d)^{-2\alpha} }^{ r_{1}^{-2\alpha}} \frac{F_{R}(s_{j}^{-1/2\alpha})-F_{R}(r_{1})}{1-F_{R}(r_{1})} \,\textmd{d} s_{j} \right] \\
&+ 4(N-1)(N-2)\sigma _{s}^{4} e^{\sigma _{z}^{2}} \left[ (R_{\textmd{W}}+d)^{-\alpha} +  \int_{(R_{\textmd{W}} + d)^{-\alpha} }^{ r_{1}^{-\alpha}} \frac{F_{R}(s_{j}^{-1/\alpha})-F_{R}(r_{1})}{1-F_{R}(r_{1})} \,\textmd{d} s_{j} \right]^{2}.
\end{split}
\end{equation}
\end{proposition*}

\begin{proof}
The achieved SIR at an arbitrary point $\textbf{x}$, given $r_{1}$, is written as
\begin{equation} \label{eq_14}
\mathtt{SIR}_{\,r_{1}} =\frac{\sigma _{s}^{2} \left|\bar{h}_{{\kern 1pt} 1} \right|^{2} r_{1}^{-\alpha } z_{1} }{\sum _{j=2}^{N}\sigma _{s}^{2} \left|\bar{h}_{j} \right|^{2} r_{j}^{-\alpha } z_{j}} =\frac{\omega _{1} {\kern 1pt} z_{1} }{\sum _{j=2}^{N}\omega _{j} {\kern 1pt} z_{j}  }   ,
\end{equation}
where the interference in the denominator is a linear combination of $N-1$ lognormal RVs $z_{j}; j=2,\cdots,N$ with coefficients $\omega _{j};j=2,\cdots,N$, which themselves are independent RVs. The key to simplifying this expression is to use the fact that, as shown in~\cite{pratesi2006generalized}, for many applications, linear combinations of lognormal random variables can be closely approximated by a single lognormal random variable. The work in~\cite{pratesi2006generalized} presents several such approximations based on a generalization of the MMA approach. For example, by matching the first and second moments, the denominator in~\eqref{eq_14} can be modeled as $\mathtt{SIR}_{\textmd{Denom}} \sim \LN(\mu _{\textmd{Denom}} , \sigma _{\textmd{Denom}})$, with
\begin{equation} \label{eq_mu_I}
\mu _{\textmd{Denom}} = \mu _{I_{r_{1}}} =2\ln (M_{1} )-0.5\ln (M_{2}),
\end{equation}
\begin{equation} \label{eq_sigma_I}
\sigma _{\textmd{Denom}}^{2} =\sigma _{I_{r_{1}}}^{2} =-2\ln (M_{1} )+\ln (M_{2}),
\end{equation}
where
\begin{equation} \label{eq_17}
M_{1} =\sum _{j=2}^{N} \Ex\{\omega _{j} \} \exp (\mu _{z_{j} } +\sigma _{z_{j} }^{2} /2),
\end{equation}
\begin{equation} \label{eq_18}
\begin{split}
M_{2} = & \sum _{j=2}^{N}\Ex\{ \omega _{j}^{2} \}  \exp {\kern 1pt} (2\mu _{z_{j} } + 2\sigma _{z_{j} }^{2} ) \\
 & +\sum _{j=2}^{N}\sum _{j'= 2 \, ; \, j'\ne j}^{N}\Ex\{ \omega _{j}^{} \} \Ex\{ \omega _{j'}^{} \} \exp ( \mu _{z_{j}} +\mu _{z_{j'} } +(\sigma _{z_{j} }^{2} +\sigma _{z_{j'} }^{2} )/2).
 \end{split}
\end{equation}

Let $s_{j} = r^{-\alpha}_{j}$. The CDF of $s_{j}$, given $r_{1}$, is obtained from $F_{R_{j} \mid \,r_{1}}(r_{j})$ as $F_{S_{j}\,\mid \,r_{1}}(s_{j}|\,r_{1}) = 1 - F_{R_{j}\,\mid \,r_{1}}(s^{-1/\alpha}_{j})$. Thus, we have $F_{S_{j}\,\mid \,r_{1}}((R_{\textmd{W}}+d)^{-\alpha}) = 0$, and $F_{S_{j}\,\mid \,r_{1}}(r^{-\alpha}_{1}) = 1$. Since the average of a random variable $Y$ can be obtained from the CDF of $Y$ as $\Ex \{Y\}=y F_{Y}(y)\mid^{b}_{a}-\int^{b}_{a} F_{Y}(y)\textmd{d}y$, where $a$ and $b$ are the values at which $F_{Y}(a)=0$ and $F_{Y}(b)=1$, we get
\begin{equation} \label{eq_19}
\begin{split}
\Ex\left\{r_{j}^{-\alpha} \mid r_{1} \right\} & = \Ex\left\{s_{j} \mid r_{1} \right\} = r^{-\alpha}_{1} -  \int _{(R_{\textmd{W}} + d)^{-\alpha} }^{ r_{1}^{-\alpha}} F_{S_{j}\,\mid \,r_{1}}(s_{j})  \,\textmd{d} s_{j} \\
& = (R_{\textmd{W}}+d)^{-\alpha} +  \int_{(R_{\textmd{W}} + d)^{-\alpha} }^{ r_{1}^{-\alpha}} \frac{F_{R}(s_{j}^{-1/\alpha})-F_{R}(r_{1})}{1-F_{R}(r_{1})} \,\textmd{d} s_{j}.
\end{split}
\end{equation}

Similarly,
\begin{equation} \label{eq_20}
\Ex\left\{r_{j}^{-2\alpha} \mid r_{1} \right\}  = (R_{\textmd{W}}+d)^{-2\alpha} +  \int_{(R_{\textmd{W}} + d)^{-2\alpha} }^{ r_{1}^{-2\alpha}} \frac{F_{R}(s_{j}^{-1/2\alpha})-F_{R}(r_{1})}{1-F_{R}(r_{1})} \,\textmd{d} s_{j}.
\end{equation}

Therefore,~\eqref{eq_17}-\eqref{eq_18} can be rewritten as~\eqref{eq_21}-\eqref{eq_22}. The integrals in~\eqref{eq_21}-\eqref{eq_22} can be easily evaluated numerically. It is worth noting that since we do not have an ordering in the interfering APs, for a given $r_{1}$, all the random coefficients $\omega_j,j=2,\cdots,N$ have equal mean and standard deviation as implied from~\eqref{eq_19}-\eqref{eq_20}. Therefore, since all the components in $\sum_{j=2}^{N} \omega _{j} z_{j}$ have equal mean and standard deviation, the MMA approach provides a good approximation for the summations of lognormals~\cite{mehta2007approximating,di2009further} (the accuracy is verified via simulations below).

The numerator in~\eqref{eq_14}, on the other hand, is a scaled lognormal RV (with a random scaling having exponential distribution), not a linear combination of lognormals; however, for the analytical tractability, we still use the MMA technique to approximate the numerator as a lognormal RV given as, $\mathtt{SIR}_{\textmd{Num}} \sim \LN(\mu_{\textmd{Num}},\sigma_{\textmd{Num}})$, with
\begin{eqnarray}
\mu_{\textmd{Num}} & = & 2\ln (\beta_{1} )-0.5\ln (\beta _{2} )=\ln (\sigma _{s}^{2} /\sqrt{2} )+\ln r_{1} ^{-\alpha }, \label{eq_15} \\
\sigma_{\textmd{Num}}^{2} & = & -2\ln (\beta_{1} )+\ln (\beta_{2} )=\ln 2+\sigma _{z}^{2}  \label{eq_16}.
\end{eqnarray}
As we will see, this approximation is remarkably accurate.

With numerator and denominator both modeled as lognormal RVs, the achieved SIR, conditioned on the connection distance $r_{1}$ between the user and the closest AP to the user, is also a lognormal random variable, $\mathtt{SIR}_{\,r_{1} } \sim \LN(\mu _{\mathtt{SIR}} ,\sigma _{\mathtt{SIR}} )$ with $\mu _{\mathtt{SIR}} =\mu _{\textmd{Num}} -\mu _{\textmd{Denom}}$ and $\sigma _{\mathtt{SIR}}^{2} =\sigma _{\textmd{Num}}^{2} +\sigma _{\textmd{Denom}}^{2}$ given in~\eqref{eq_mu_SIR} and~\eqref{eq_sigma_SIR}, respectively. Having found an approximate distribution of the SIR as a lognormal random variable, the conditional SIR coverage probability (conditioned on the distance $r_{1}$) is
\begin{equation} \label{eq_23}
\Px\{ \mathtt{SIR}_{\,r_{1}} > T \} = Q\left( \frac{\ln T - \mu _{\mathtt{SIR}}}{\sigma _{\mathtt{SIR}}} \right),
\end{equation}
Finally, by averaging over distance $r_{1}$, we obtain the result.
\end{proof}

Correspondingly, from~\eqref{eq_4}, the average user capacity coverage probability follows as
\begin{equation} \label{eq_25}
\begin{split}
\textmd{CP}_{C}^{\textmd{\,avg}} (N, d, C_{0}, \alpha, \sigma_L) = & \Px\{ C > C_{0} \} \\
= & \int_{0}^{R_{\textmd{W}} + d } Q\left( \frac{\ln (2^{C_{0}/N} -1) - \mu _{\mathtt{SIR}}}{\sigma _{\mathtt{SIR}}} \right) \frac{d F_{R_{1}}(r_{1})}{d\,r_{1}} \,d\,r_{1}.
\end{split}
\end{equation}

Finally, using the fact that for a positive random variable $X$, $\Ex \{X\} = \int_{t \geq 0} \Px \{X >t\}\,\textmd{d}\,t$, we can use the capacity coverage probability to obtain the ergodic capacity as
\begin{equation} \label{eq_26}
{C}_{\textmd{\,ergodic}} (N, d,\alpha, \sigma_L)=\int_{C_0 \geq 0} \textmd{CP}_{C}^{\textmd{\,avg}} (N, d, C_{0}, \alpha, \sigma_L)\, \textmd{d}C_{0}.
\end{equation}

Although~\eqref{eq_24}-\eqref{eq_26} do not explicitly show the resulting dependence on the AP density, $\lambda$, they can be rewritten as functions of $\lambda$ by using the substitution $N = \pi R_{\textmd{W}}^{2} \lambda$. As we will see in the next section, simulations show that in the asymptotic case of a large number of APs (large AP density), the values obtained using~\eqref{eq_24}-\eqref{eq_26} converge to those of an infinite network (or dense network) given in~\cite{andrews2011tractable,dhillon2012modeling}.

\subsection{Network with thermal noise: $\sigma^2_n \neq 0$}

When including thermal noise, the instantaneous achievable SINR is given by~\eqref{eq_3}. The noise term, $\sigma^2_n$, in the denominator of~\eqref{eq_3} is deterministic, and therefore it is (trivially) a lognormal RV with mean $\mu _{\sigma _{n}^{2} } =\ln \sigma _{n}^{2} $ and zero variance ($\sigma _{\sigma _{n}^{2} }^{2} =0$) for the associated Gaussian RV, $\ln (\sigma _{n}^{2} )$. Now, the denominator of SINR becomes the addition of two lognormal RVs $\sigma _{n}^{2} \sim \LN(\mu _{\sigma _{n}^{2}} ,{\kern 1pt} {\kern 1pt} \sigma _{\sigma _{n}^{2} } )$ and $I_{r_{1}} \sim \LN(\mu _{I_{r_{1}}} ,\sigma _{I_{r_{1}}})$. Consequently, as we apply the MMA technique in order to approximate the addition of these two lognormal RVs with another lognormal RV, the denominator of SINR can be modelled as
\begin{equation} \label{eq_27}
\mathtt{SINR}_{\textmd{Denom}} \sim \LN(\mu _{\sigma _{n}^{2} +I_{r_{1}}} , \sigma _{\sigma _{n}^{2} +I_{r_{1}}})
\end{equation}
with $\mu _{\sigma _{n}^{2} +I_{r_{1}}} =2\ln (\bar{M}_{1} )-0.5\ln (\bar{M}_{2})$, and $\sigma _{\sigma _{n}^{2} +I_{r_{1}}}^{2} =-2\ln (\bar{M}_{1} )+\ln (\bar{M}_{2} )$, where
\begin{eqnarray} \label{eq_28}
\bar{M}_{1} & = & e^{\ln \sigma _{n}^{2} } +e^{\mu _{I_{r_{1}}} + \, \sigma _{I_{r_{1}} }^{2} /\, 2}, \\
\bar{M}_{2} & = & e^{2\ln \sigma _{n}^{2}} +e^{2\mu _{I_{r_{1}}} + 2\sigma _{I_{r_{1}} }^{2}} +2\sigma _{n}^{2} e^{\mu _{I_{r_{1}}} + \sigma _{I_{r_{1}}}^{2} /2}. \label{eq_29}
\end{eqnarray}

Therefore, as before, since both the numerator and denominator are approximated as lognormal RVs, the achieved SINR is a lognormal RV for a given $r_{1}$. We get, $\mathtt{SINR}_{\,r_{1} } \sim \LN(\mu _{\mathtt{SINR}} ,\sigma _{\mathtt{SINR}} )$ with $\mu _{\mathtt{SINR}} =\mu _{\textmd{Num}} -\mu _{\textmd{Denom}}$ and $\sigma _{\mathtt{SINR}}^{2} =\sigma _{\textmd{Num}}^{2} +\sigma _{\textmd{Denom}}^{2}$. Therefore, the SINR/capacity coverage probability averaged over different realizations of AP locations is obtained from~\eqref{eq_24}-\eqref{eq_26} by substituting  $\mu _{\mathtt{SIR}}$ and $\sigma _{\mathtt{SIR}}$ with $\mu _{\mathtt{SINR}}$ and $\sigma _{\mathtt{SINR}}$, respectively.

\subsection{Worst-case Point}

In Appendix A we show that, for small values of noise variance, the worst-case SINR and hence user capacity occurs at the center of the circular area $\textmd{W}$, i.e., when $d = 0$. In this case, some of the expressions provided in the previous sub-sections can be simplified. With $d = 0$, the CDF of the distance between a user at the center to an arbitrary AP randomly located in the circular region is $F_{R\,}(r) = r^2 / R^{\,2}_{\textmd{W}}$. Therefore, the CDF and PDF of $r_{1}$ become
\begin{eqnarray}
F_{R_{1}}(r_{1}) & = & 1-\left(1-\frac{r_{1}^{2} }{R_{\textmd{W}}^{\,2}}\right)^{N}, \label{eq_30} \\
f_{R_{1}}(r_{1})& = & \frac{d F_{R_{1}}(r_{1})}{dr_{1}} =\frac{2Nr_{1}}{R_{\textmd{W}}^{\,2}} \left(1-\frac{r_{1}^{2}}{R_{\textmd{W}}^{\,2}} \right)^{N-1} ~~, ~0 \le r_{1} \le R_{\textmd{W}}. \label{eq_31}
\end{eqnarray}

Now, for a given $r_{1} $, the distribution of the $N-1$ interfering APs in the area between circles centred at the origin with radii $r_{1} $ and $R_{\textmd{W}}$, denoted as $\textmd{B}$, is that of $(N-1)$ i.i.d. random points $(x_{j} ,y_{j} ){\kern 1pt}  , j=2,\cdots,N$, uniformly distributed in $\textmd{B}$ with common distribution $f_{x_{j} ,y_{j} } (x_{j} ,y_{j} )=1/S(\textmd{B})=1/\pi (R_{\textmd{W}}^{2} -r_{1}^{2} ){\kern 1pt} {\kern 1pt} $ expressed in Cartesian coordinates. With the change of variable $x_{j} =r_{j} \cos \theta _{j} $ and $y_{j} =r_{j} \sin \theta _{j} $, and then integrating over the resulting uniform distribution in $\theta _{j} $, $0\le \theta _{j} \le 2\pi $, the distance PDF of an individual AP location, for a given value of $r_1$, is given by
\begin{equation} \label{eq_32}
f_{R_{j} \mid r_{1}} (r_{j})=\frac{2{\kern 1pt} r_{j} }{R_{\textmd{W}}^{\,2}-r_{1}^{2}} {\kern 1pt} {\kern 1pt} {\kern 1pt} {\kern 1pt} ,~~ r_{1} \le r_{j} \le R_{\textmd{W}},\, \, \, \, j=2,\cdots, N,
\end{equation}
and zero elsewhere. Therefore, we have
\begin{equation} \label{eq_33}
\Ex\left\{r_{j}^{-\alpha} | r_{1} \right\}=\int _{r_{1} }^{R_{\textmd{W}} }r_{j}^{-\alpha } f_{R_{j} \mid r_{1}} (r_{j})\,dr_{j} =\frac{2(r_{1}^{-\alpha +2} -R_{\textmd{W}}^{-\alpha +2} )}{(\alpha -2)(R_{\textmd{W}}^{\,2} -r_{1}^{2} )},
\end{equation}
\begin{equation} \label{eq_34}
\Ex\left\{ r_{j}^{-2\alpha } | r_{1} \right\}=\int _{r_{1} }^{R_{\textmd{W}} }r_{j}^{-2\alpha } f_{R_{j} \mid r_{1}} (r_{j}) {\kern 1pt} dr_{j} =\frac{2(r_{1}^{-2\alpha +2} -R_{\textmd{W}}^{-2\alpha +2} )}{(2\alpha -2)(R_{\textmd{W}}^{\,2} -r_{1}^{2} )}.
\end{equation}
and so Eqs.~\eqref{eq_17}-\eqref{eq_18} simplify to
\begin{eqnarray}
M_{1} & = & \frac{2(N-1)\sigma _{s}^{2} e^{\sigma _{z}^{2} /2}}{\alpha -2} \left(\frac{r_{1}^{-\alpha +2} -R_{\textmd{W}}^{-\alpha +2} }{R_{\textmd{W}}^{\,2} -r_{1}^{2} } \right), \\
M_{2} &= & \frac{4(N-1)\sigma _{s}^{4} e^{2\sigma _{z}^{2} } }{2\alpha -2} \left(\frac{r_{1}^{-2\alpha +2} -R_{\textmd{W}}^{-2\alpha +2} }{R_{\textmd{W}}^{\,2} -r_{1}^{2} } \right) \nonumber \\
& & \hspace*{0.3in} +\frac{4(N-1)(N-2)\sigma _{s}^{4} e^{\sigma _{z}^{2} } }{(\alpha -2)^{2} } \left(\frac{r_{1}^{-\alpha +2} -R_{\textmd{W}}^{-\alpha +2} }{R_{\textmd{W}}^{\,2} -r_{1}^{2} } \right)^{2}.
\end{eqnarray}
The rest of the expressions remain unchanged.

It is worth noting that, in general, no closed-form expression is available for the integrations in~\eqref{eq_24}-\eqref{eq_26} as a function of $\alpha$. However, a tractable analysis is possible for specific integer values of $\alpha$. As an example for the popular value of $\alpha=4$~\cite{andrews2011tractable}, we show in Appendix B that an accurate analytical approximation for the worst-case ergodic user capacity in an interference-limited network is obtained as (for $N>2$)
\begin{equation} \label{eq_35_new}
 \begin{split}
 C_{\textmd{\,ergodic}}^{\textmd{\,worst}} \approx \frac{N}{\ln 2} & \bigg( \gamma +\frac{1}{2N} +\ln \frac{N(N-2)^{1/2} }{(N-1)^{3/2} } + \frac{1}{2} \ln \frac{N-2}{N-1} +\frac{1}{2} \ln \bigg[1+\frac{2e^{\sigma _{z}^{2} } }{3(N-2)} \bigg] \\
 & + \left. \bigg(\big(1+\frac{2e^{\sigma _{z}^{2} } }{3(N-2)} \big)^{N} -1 \bigg) \bigg(\ln \bigg[\frac{1+1.5e^{-\sigma _{z}^{2} } (N-2)}{N-1} \bigg]-\frac{1}{2(N-1)} \bigg)\right)
 \end{split}
\end{equation}
where $\gamma =0.578$ is the Euler-Mascheroni constant~\cite{havil2003exploring} and $\sigma _{z}^{2} =(0.1\ln 10)^{2} \sigma _{L}^{2}$. Using similar approximations, closed-form expressions can also be obtained for other integer values of $\alpha \geq 2$.

Further analysis is possible to investigate how quickly the performance approaches the worst case as $d$ becomes small compared to $R_{\textmd{W}}$. We show in Appendix C that, for $d \ll R_{\textmd{W}}$, the change in averaged SIR (expressed in dB) compared to the worst case obtained at the centre of an interference-limited network with $\alpha = 3.87$, can be closely approximated by a $3^{\textrm{rd}}$ order polynomial function of $d$.

\section{Numerical Results}\label{sec:Sims}

In general, the CP results obtained in the previous section depend on various parameters of the network. In this section, we simulate two typical examples of $\alpha =3$, and $\alpha =3.87$. Results can also be given for other values of $\alpha$, but the example provided here is sufficient to illustrate the approach. A circular finite-area downlink with radius $R_{\textmd{W}} =1~\textmd{km}$ is considered and the transmit signal power of each AP corresponds to $\sigma _{s}^{2} =20~\textmd{dBm}$; we note that the parameter values in the simulations are just for illustration purposes, and any other values only scale the results.


\subsection{Interference-limited network: $\sigma^2_n = 0$}

In Section~\ref{sec:SINR}, a series of approximations were used to derive the expressions for the SIR/SINR coverage probability (or capacity coverage probability). Thus, it is important to validate the approximations. The 2D plot of the averaged SIR coverage probability, for the target SIR of $0~\textmd{dB}$, obtained within the area of the circle is given in Fig.~\ref{Fig_2D}. In this example, we set $\lambda = 1~\textmd{APs/km}^2$, $\alpha =3.87$, and $\sigma_L = 6~\textmd{dB}$. At this low density, the network is interference limited if the transmit power is high. Figure~\ref{Fig_2D}-(a) presents results obtained from the analysis developed here while Fig.~\ref{Fig_2D}-(b) presents the results from Monte Carlo simulations. The analysis accurately approximates the exact results obtained from simulations within an average error of 5\%. The difference between~\ref{Fig_2D}-(a) and \ref{Fig_2D}-(b) is due to the discretization and number of iterations chosen for the simulations.
\begin{figure}[t]
 \begin{center}
 \subfigure[Results from analysis]{\includegraphics[width = .471\textwidth]{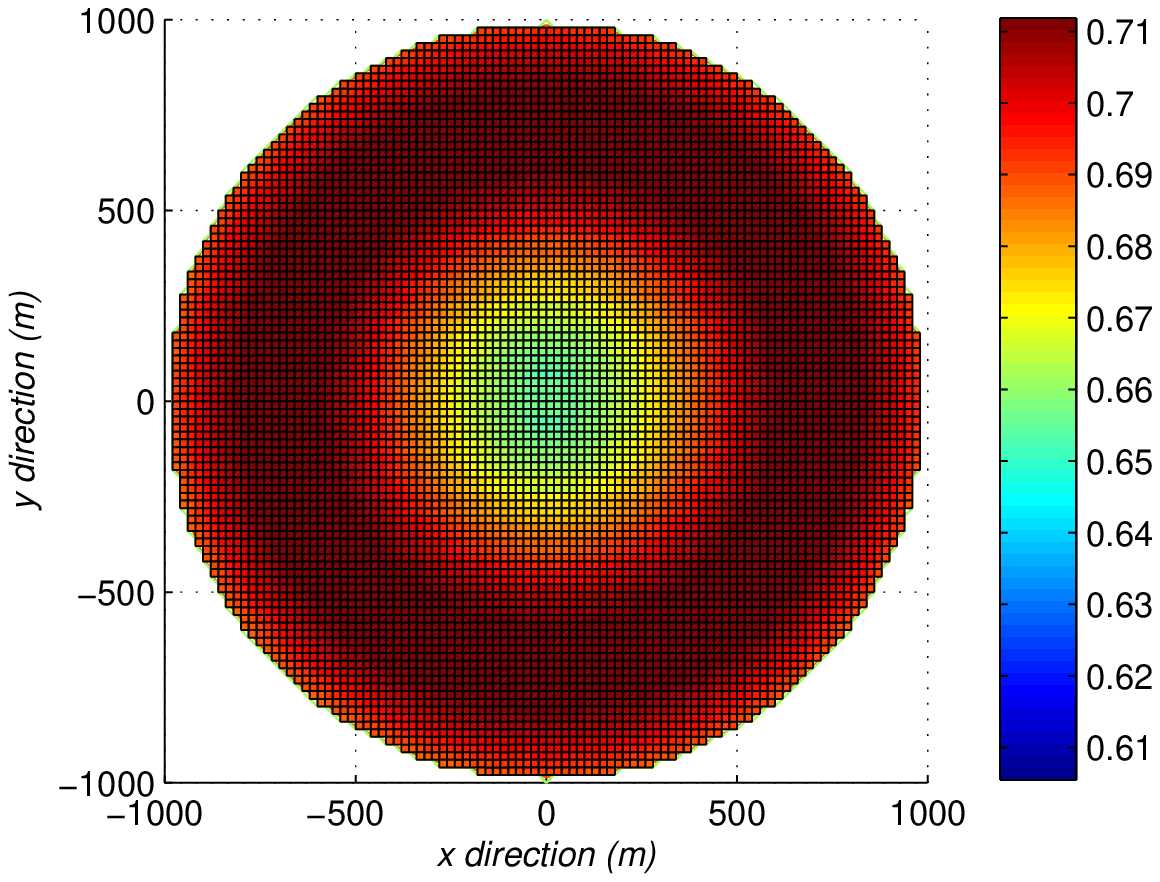}}
 \subfigure[Results from simulations]{\includegraphics[width = .471\textwidth]{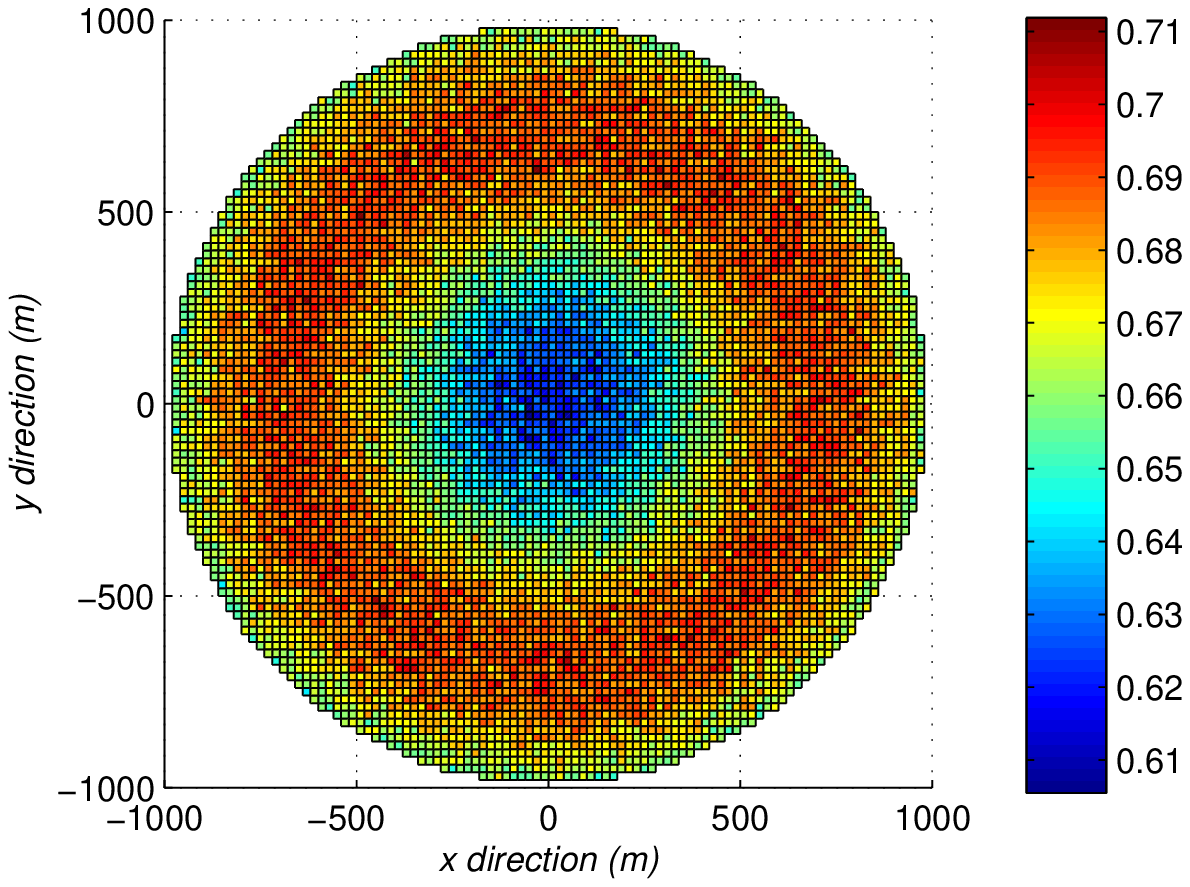}}
 \caption{SIR coverage probability (for the target SIR of $0~\textmd{dB}$) in a circular finite-area interference-limited network with $\lambda = 1~\textmd{APs/km}^2$, $\alpha =3.87$, and $\sigma_L = 6~\textmd{dB}$.}
 \label{Fig_2D}
 \end{center}
\end{figure}
Since, by averaging over AP locations, the coverage probability in a circular area is independent of angle, it is enough to evaluate the results along any radial line. Figure~\ref{Fig_3} illustrates the approximate averaged SIR coverage probability, for the target SIR of $0~\textmd{dB}$, along a radial line of the circle. The results are illustrated for the two examples of $\sigma_L = 0~\textmd{dB}$ (no shadowing) and $\sigma_L = 6~\textmd{dB}$. The figure plots the results for different values of AP density. The results from Monte Carlo simulations are also included in the figure (the dotted lines in the figure).
\begin{figure}[t]
 \begin{center}
 \subfigure[$\alpha = 3$, $\sigma_L = 0~\textmd{dB}$]{\includegraphics[width = .49\textwidth]{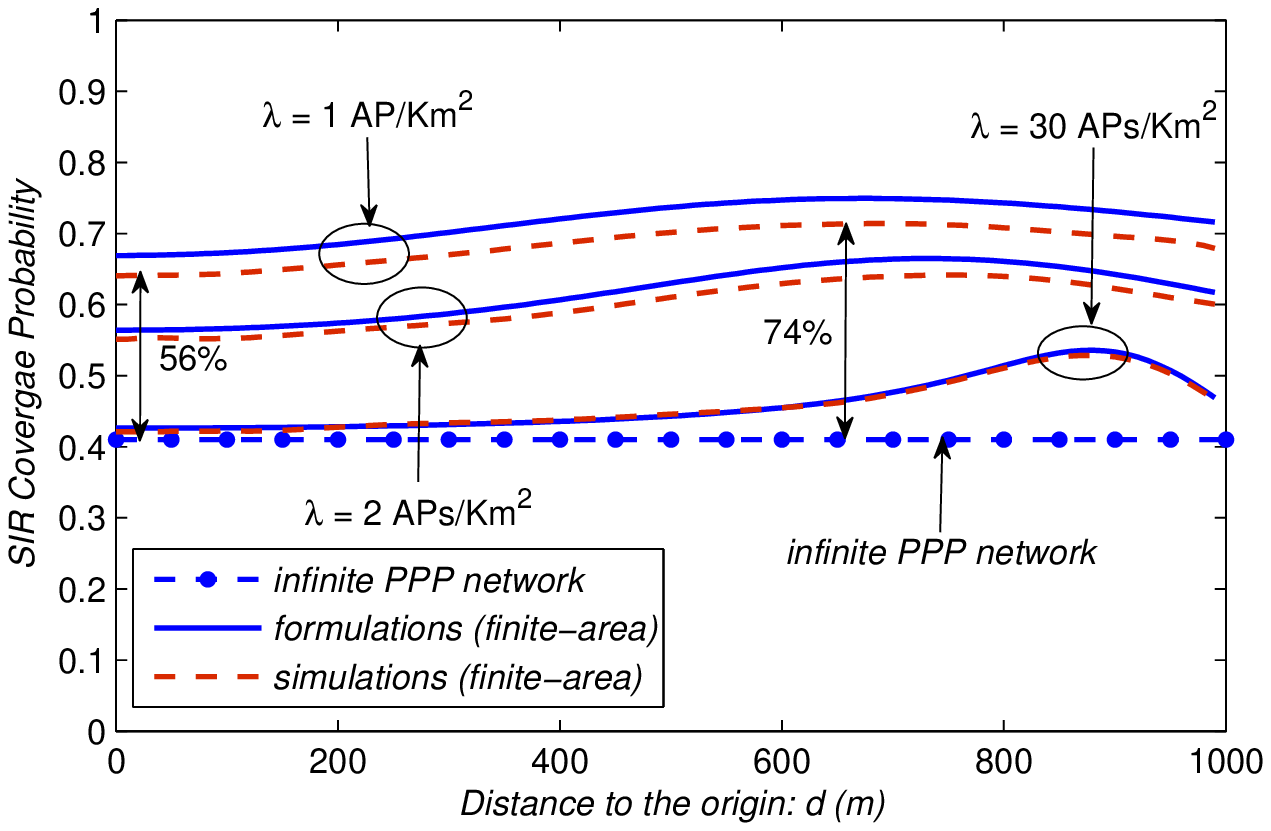}}
 \subfigure[$\alpha = 3$, $\sigma_L = 6~\textmd{dB}$]{\includegraphics[width = .49\textwidth]{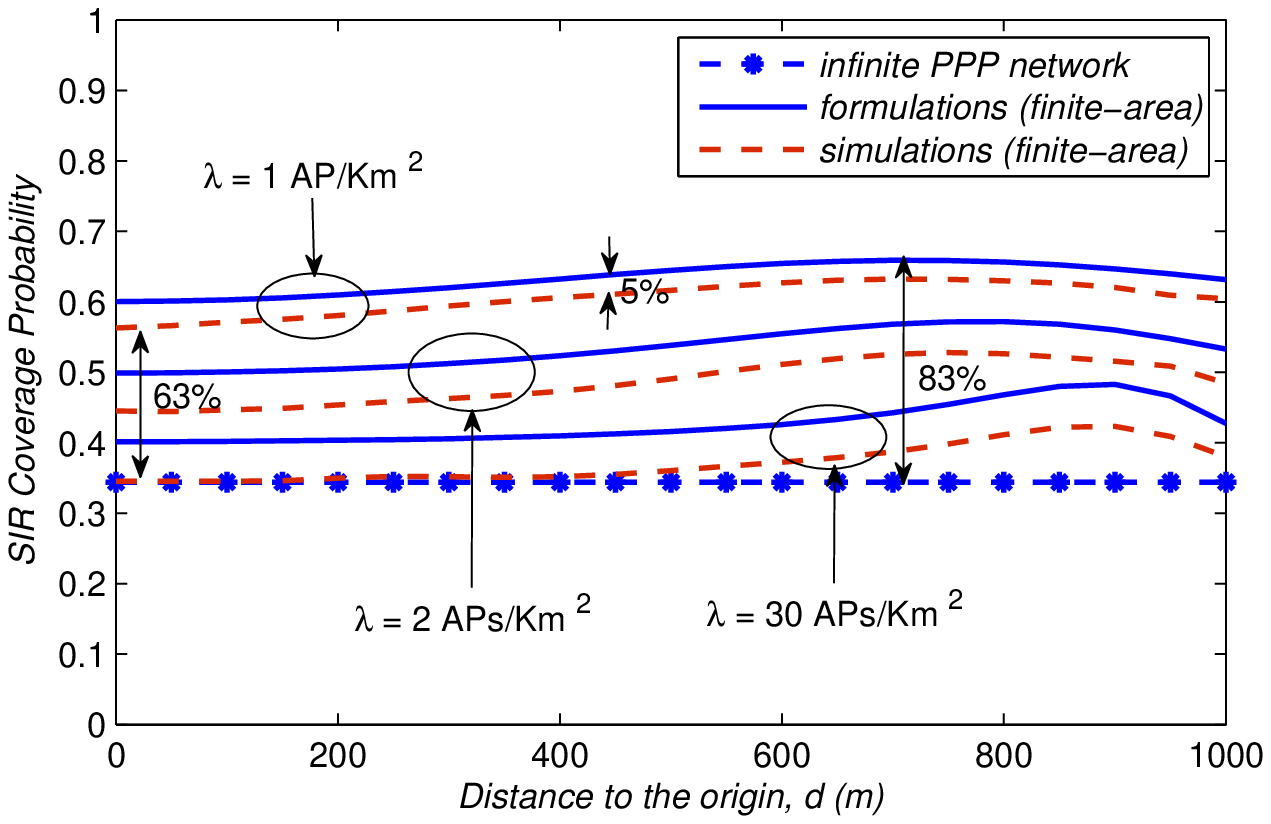}}
 \subfigure[$\alpha = 3.87$, $\sigma_L = 0~\textmd{dB}$]{\includegraphics[width = .49\textwidth]{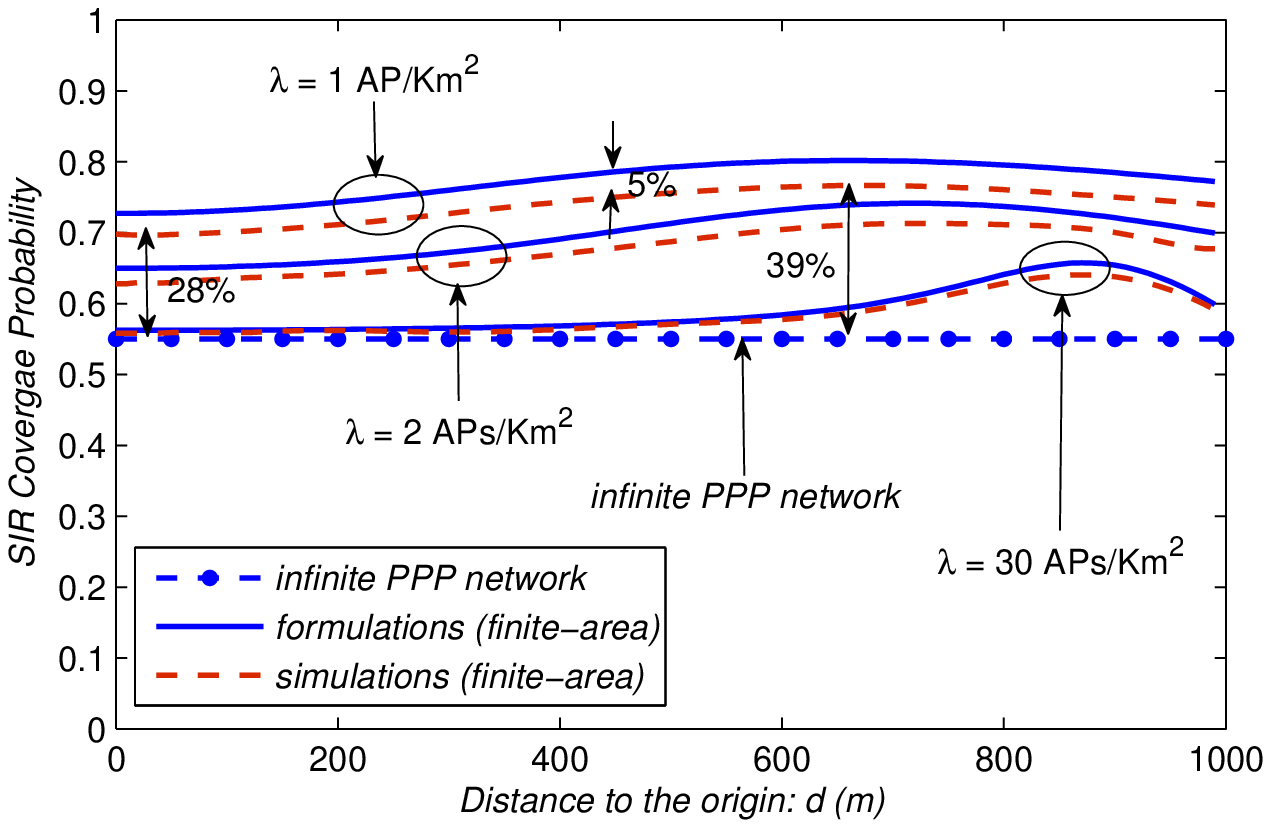}}
 \subfigure[$\alpha = 3.87$, $\sigma_L = 6~\textmd{dB}$]{\includegraphics[width = .49\textwidth]{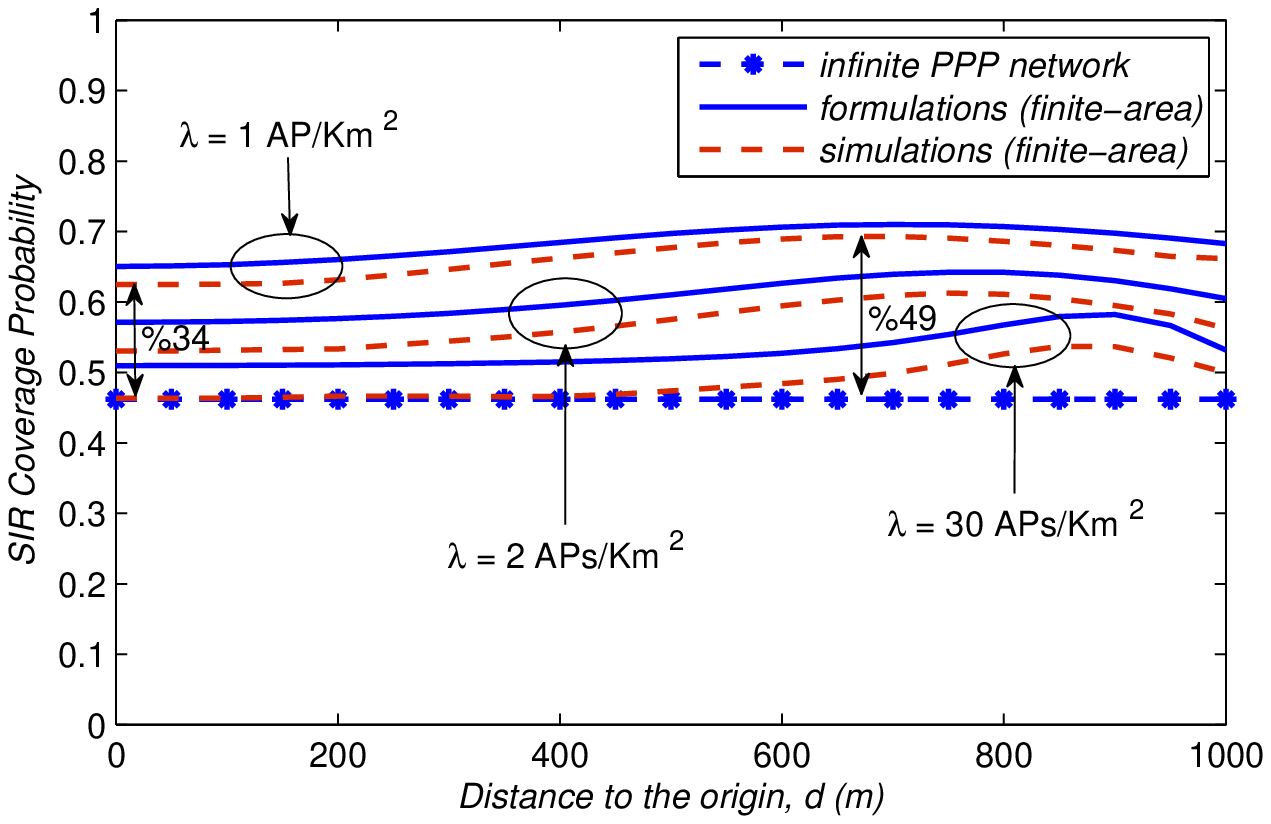}}
 \caption{SIR coverage probability (for the target SIR of $0~\textmd{dB}$) along the radius of the circular finite-area interference-limited network with AP densities $\lambda = 1~\textmd{APs/km}^2$, $\lambda = 2~\textmd{APs/km}^2$ and $\lambda = 30~\textmd{APs/km}^2$ for: a) $\alpha =3$, $\sigma_L = 0~\textmd{dB}$; b) $\alpha =3$, $\sigma_L = 6~\textmd{dB}$; c) $\alpha =3.87$, $\sigma_L = 0~\textmd{dB}$; d) $\alpha =3.87$, $\sigma_L = 6~\textmd{dB}$.}
 \label{Fig_3}
 \end{center}
 \end{figure}

As is clear from the figure, the analysis of the previous section capture the behavior of the system quite well. For low-density networks ($\lambda = 1~\textmd{AP/km}^2$) the analytical results are within 5\% of the simulated results while the error reaches 9.5\% for highly dense networks ($\lambda = 30~\textmd{AP/km}^2$) under moderate to high values of shadowing standard deviation.

Of importance is the significant differences in the coverage probability between low-density networks and the asymptotic case available in the literature (the infinite PPP network curve). As has been reported earlier~\cite{vijayandran2012analysis},  in an interference-limited network, the results of an infinite network underestimates the SIR coverage probability for small to moderate values of AP density. In particular, in this example in Fig.~\ref{Fig_3} with $\sigma_L = 6~\textmd{dB}$, the SIR coverage probability at the worst-case point (at the origin) of a circular network with $\lambda = 1~\textmd{APs/km}^2$ and $\alpha = 3.87$, outperforms that of dense network by 34\%. The improvement in the SIR coverage probability increases to 63\% under the PLE of $\alpha = 3$. It is characterizing this difference that has motivated this paper.

It is worth noting that the improvement in SIR in a low-density network depends heavily on the PLE, shadowing standard deviation and the location under consideration. For example, for the target SIR of $0~\textmd{dB}$ under no shadowing, an increase of at least 28\% in SIR coverage probability is obtained in an interference-limited network with $\alpha = 3.87$ and $\lambda = 1 ~\textmd{APs/km}^2$ as compared to a dense network with $\lambda = 30~\textmd{APs/km}^2$ (56\% improvement with $\alpha = 3$). It is also interesting to note that for users near the edge (large $d$), even this curve deviates from the simulation results for the high-density network; this is because infinite-area PPP networks inherently cannot account for edge effects.

As is seen from Figs.~\ref{Fig_2D} and~\ref{Fig_3}, for a given AP density, the circular finite-area network experiences a peak in SINR CP at a certain distance $d$ from the center. In addition, for small values of noise variance, the worst-case SINR and/or user capacity occurs at the center of circular region $\textmd{W}$. Please refer to Appendix A for the explanation of these behaviours. In particular, the worst-case point is of particular interest in parametric studies for network design since it can be directly related to a coverage constraint. Therefore, in most of the simulations below, we focus on the worst-case point.

We further justify the accuracy of the presented formulations for different sizes of the finite-area via Fig.~\ref{Fig_SIR_radius}. The figure illustrates the worst SIR coverage probability (for the target SIR of $0~\textmd{dB}$) in a low-density and a highly-dense network for the example of $\sigma_L = 6~\textmd{dB}$. Here, the accuracy of the presented formulations decreases with the size of the finite-area in a low-density network and the error reaches $9.5\%$  of the simulation results in large networks. In addition, Fig.~\ref{Fig_SIR_radius} illustrates two interesting behaviors. First, in a highly-dense interference-limited network, the SIR performance does not change with the size of the network. This effect reflects the fact the results of a highly dense network can be closely approximated by the results obtained in an infinite-area network - a result which was also reported earlier in~\cite{mao2012towards}. Second, the sensitivity of the network to the AP density decreases with the size of the network to such an extent that a large network becomes insensitive to the AP density. This effect matches the results previously obtained in~\cite{andrews2011tractable} that the SIR performance in an interference-limited infinite-area PPP network does not depend on the AP density.
\begin{figure}[t]
 \begin{center}
 \includegraphics[width = .5\textwidth]{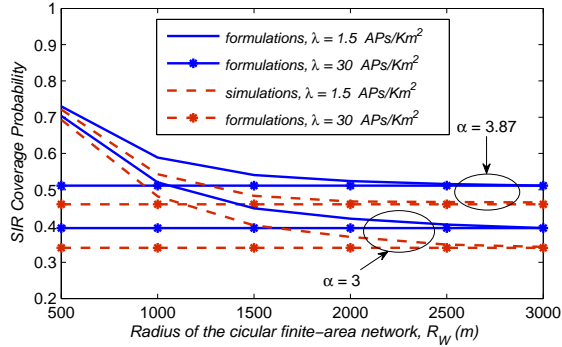}
 \caption{The effect of network radius size on the worst SIR coverage probability. The results are shown for a circular interference-limited network with $\lambda = 1.5~\textmd{AP/km}^2$ and $\lambda = 30~\textmd{APs/km}^2$ under two PLEs of $\alpha = 3$ and $\alpha = 3.87$ and $\sigma_L = 6~\textmd{dB}$.}
 \label{Fig_SIR_radius}
 \end{center}
 \end{figure}

\subsection{Network with thermal noise: $\sigma^2_n \neq 0$}

The previous results were for an interference-limited network where we ignored thermal noise. We  next  determine  the  accuracy  of  the analysis with respect to the noise variance. Figure~\ref{Fig_5} shows the effect of noise variance on the SINR coverage probability obtained at the centre of a low-density circular network for $\alpha = 3.87$. As expected, the SINR coverage probability degrades with an increase in the noise variance.
\begin{figure}[t]
 \begin{center}
 \subfigure[$\sigma_L = 0~dB$]{\includegraphics[width = .49\textwidth]{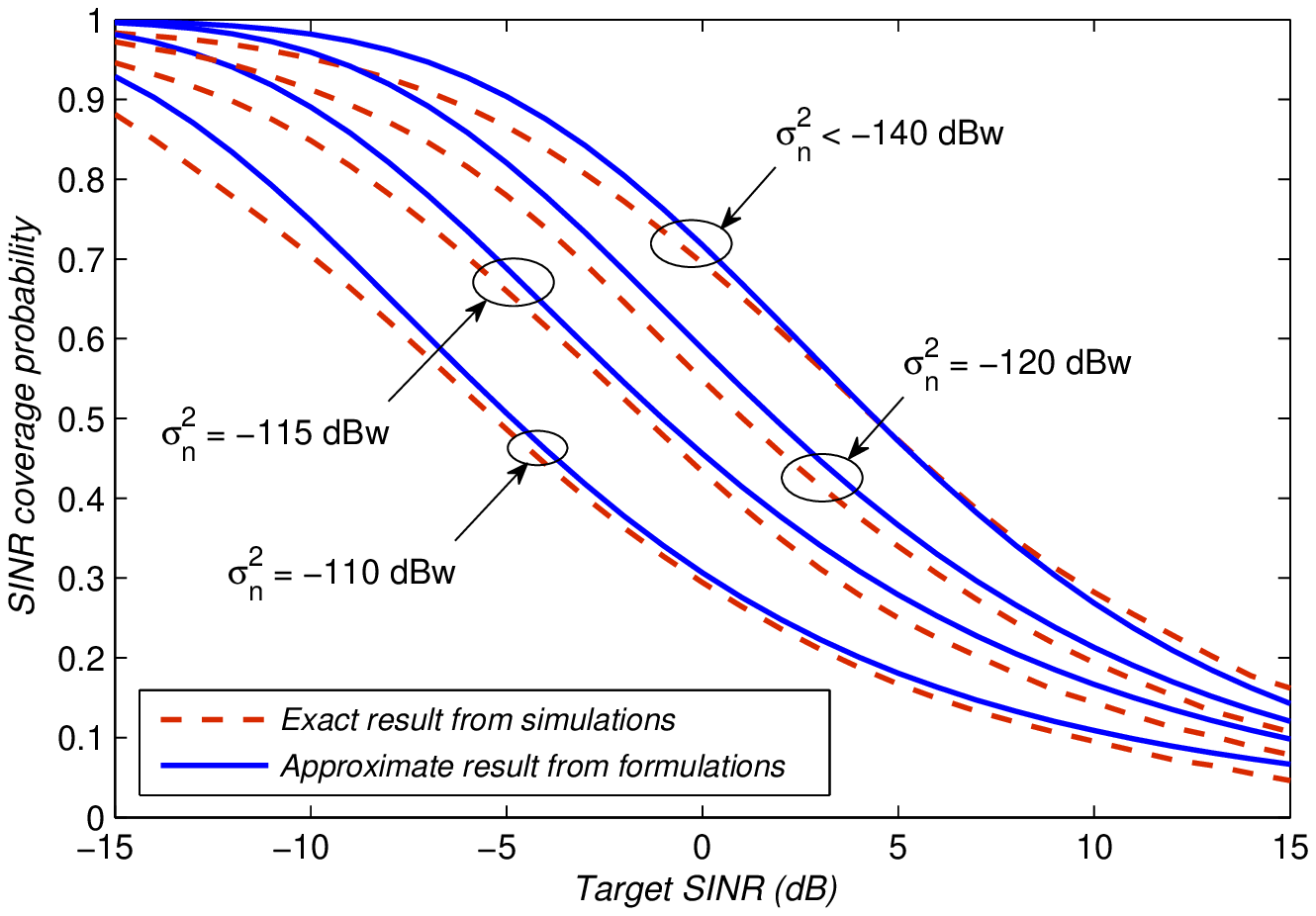}}
 \subfigure[$\sigma_L = 6~dB$]{\includegraphics[width = .49\textwidth]{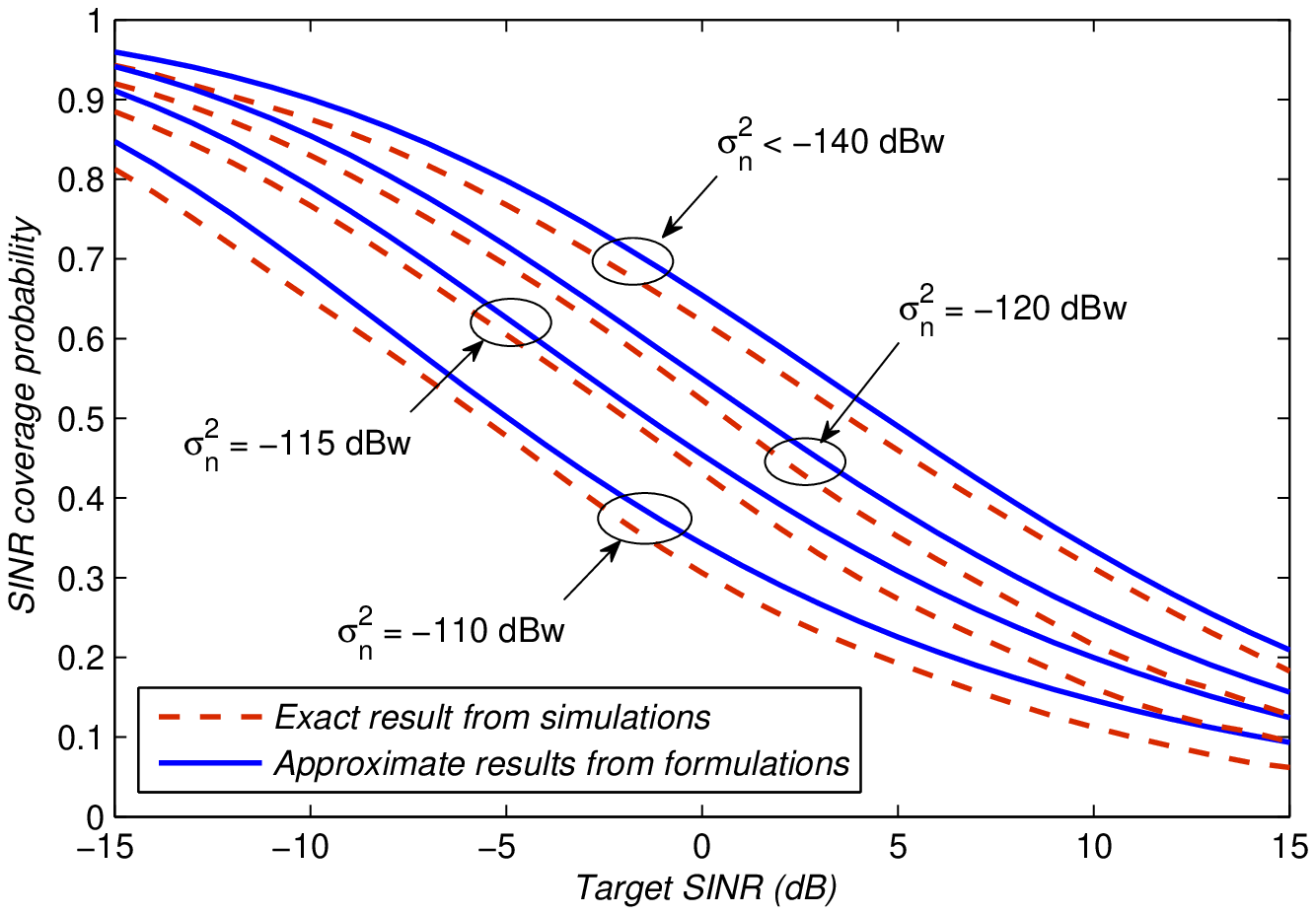}}
 \caption{The effect of noise variance on the SINR coverage probability of a low-density finite-area network with AP density $\lambda = 1~\textmd{APs/km}^2$ with PLEs of $\alpha =3.87$ for: a) $\sigma_L = 0~dB$ ; b) $\sigma_L = 6~dB$}
 \label{Fig_5}
 \end{center}
\end{figure}
The results from formulations are within 9\% for all values of noise variances. For typical values of noise variance in practice ($\sigma _{n}^{2} \le -100{\kern 1pt}~\textmd{dBm}$) the accuracy of the presented formulation is within 5\%. Importantly, in all cases, the analysis captures the behavior of the system. Therefore, the presented $\textmd{CP}_{\mathtt{SINR}}^{\textmd{\,avg}}$ expression is accurate for practical values of $\sigma_{n}^{2}$.

The effect of AP density on the SINR CP (for the target SINR of $0~\textmd{dB}$) is illustrated in Fig.~\ref{Fig_6} for different values of noise variances. The dotted lines in the figure correspond to the SINR CP obtained at the worst-case point and the solid lines are the corresponding results for the maximum\footnote{In general, no closed-form analytical expression is available for the maximum of SINR CP from the presented analytical formulations, so we use computer simulations to compute the maximum achievable SINR CP along the radius of the circle using the integral expression in~\eqref{eq_24}.} of achievable SINR CP within the finite-area.
\begin{figure}[t]
 \begin{center}
 \subfigure[$\sigma_L = 0~dB$]{\includegraphics[width = .49\textwidth]{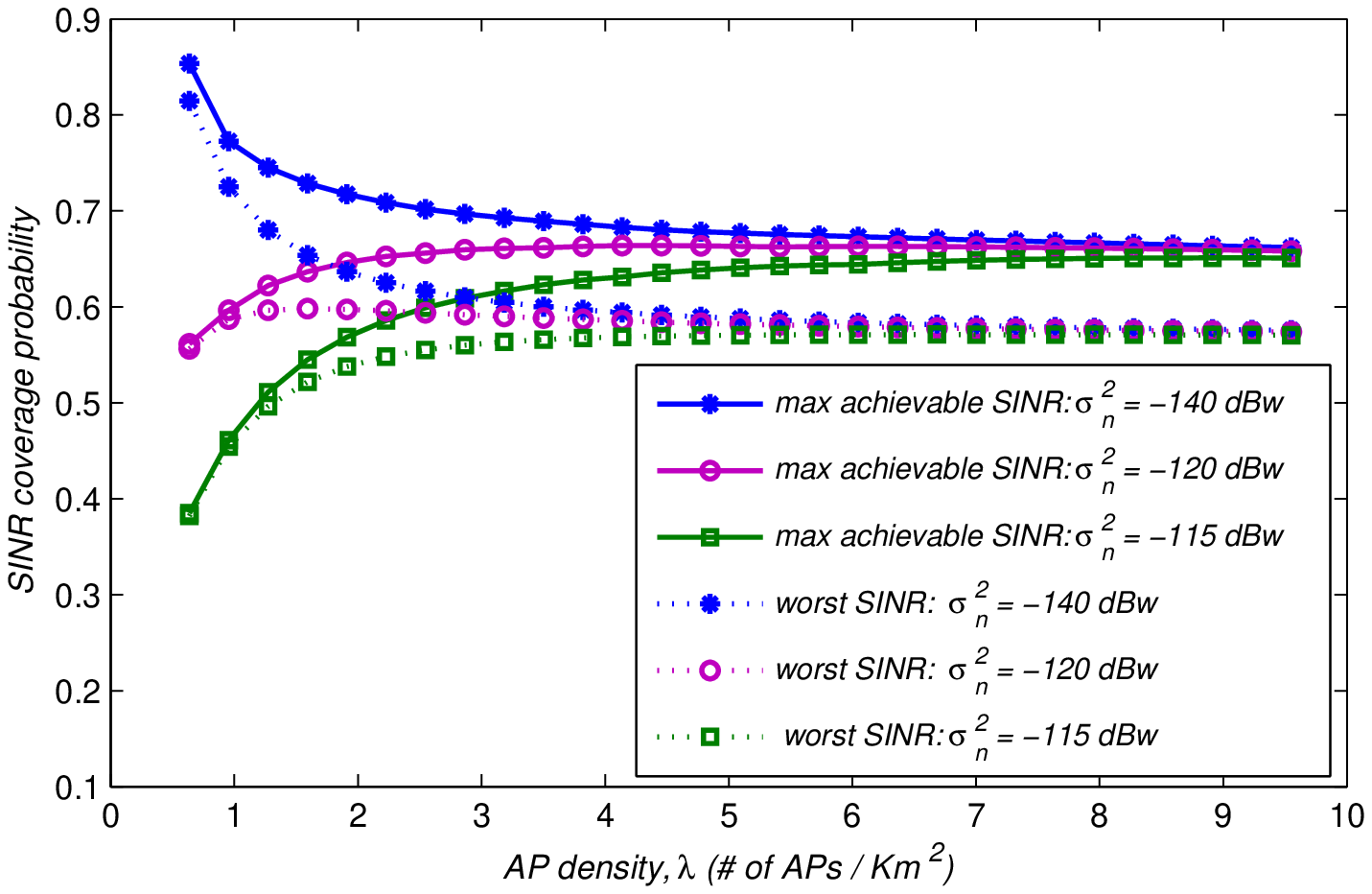}}
 \subfigure[$\sigma_L = 6~dB$]{\includegraphics[width = .49\textwidth]{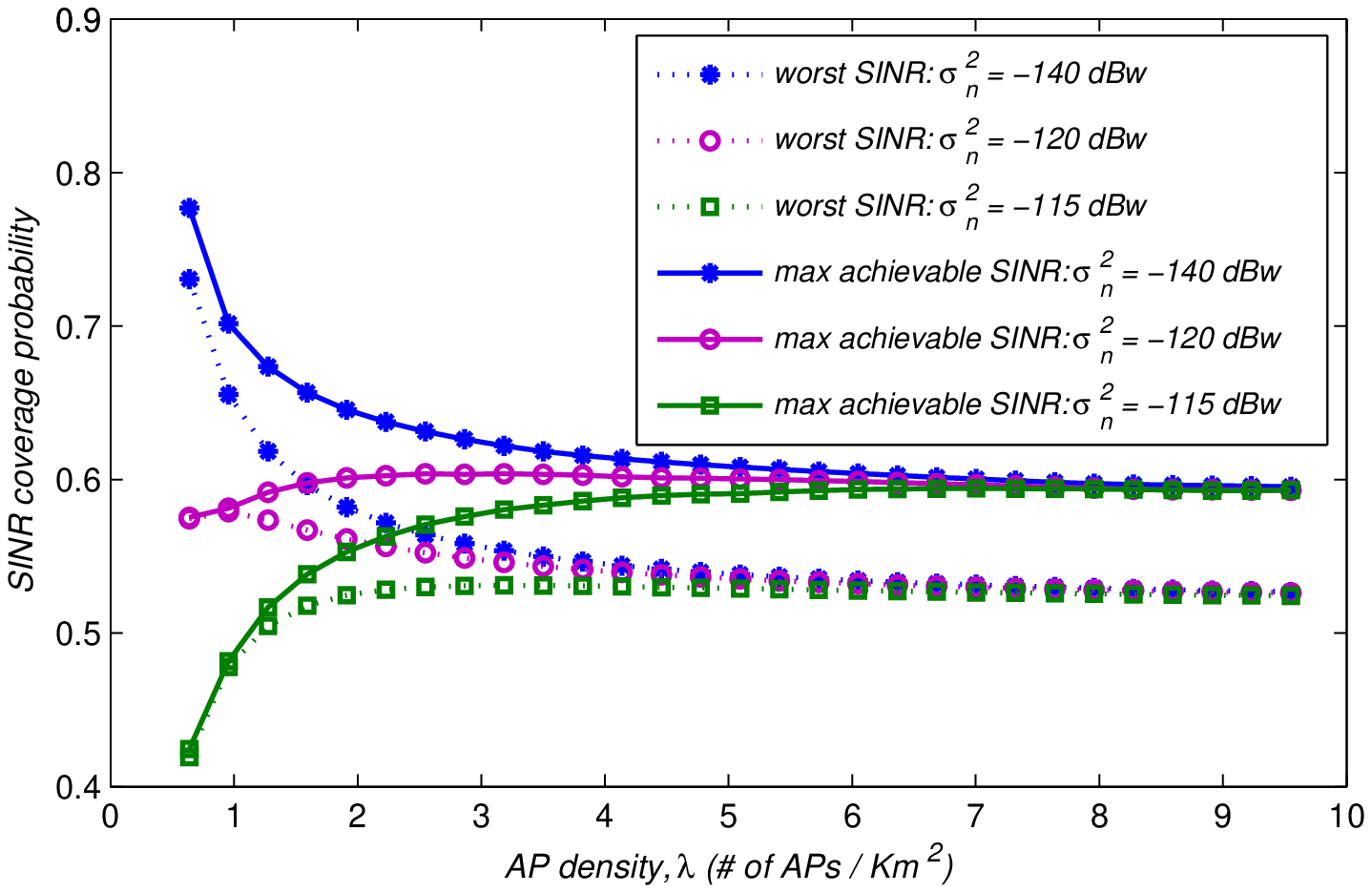}}
 \caption{The effect of AP density on the worst and maximum achievable SINR coverage probability (for the target SINR of $0~\textmd{dB}$) in a finite-area network for different values of noise variances with PLEs of $\alpha =3.87$ and: a) $\sigma_L = 0~dB$ ; b) $\sigma_L = 6~dB$}
 \label{Fig_6}
 \end{center}
\end{figure}
The different trends in SINR CP under different values of noise variance can be explained as follows.

In general, the increase in AP density (or $N$) causes the distance PDF $f_{R_{1}}(r_{1})$ to become narrower (as seen in Appendix A, in Fig.~\ref{Appendix_I_Fig_1}-(b)). As a result, the average received signal power increases with $\lambda$. The interference power also increases with $\lambda$. In a noisy network with moderate to high values of noise variance, the noise power dominates the interference power. Therefore, the increase in signal power causes the SINR CP to increases with $\lambda$, to the extent that, for large AP densities the interference power dominates the noise power. As a result, the behavior of the system under consideration converges to that of an interference-limited network. On the other hand, in an interference-limited network (very small values of noise variance), the interference dominates the thermal noise for any $\lambda$. It turns out that, the impact of interference power is more than that of signal power causing the SINR CP to degrade with $\lambda$. Finally, when increasing $\lambda$, the SINR CP converges to that in the infinite network case, which is also interference limited.

The effect of AP density on the SINR performance of the finite-area network can be further investigated by defining the ``transmit SNR" as the ratio of the transmit power to the noise variance, $\mathtt{SNR}_{\,t} =\sigma _{s}^{2} /\sigma _{n}^{2}$. Figure~\ref{Fig_color_contour} illustrates the contour plot and the color plot of the SINR CP (for the target SINR of $0~\textmd{dB}$) obtained based on different values of transmit SNR and AP density in a finite-area network with $\alpha =3.87$ and $\sigma_L = 6~\textmd{dB}$. As is seen from the contour plot in Fig.~\ref{Fig_color_contour}-{a}, there exists a very small range of $\mathtt{SNR}_t$ ($108.5~\textmd{dB} \lesssim \mathtt{SNR}_t \lesssim 109.5~\textmd{dB}$), for which two AP densities would yield the same SINR CP for a chosen value of $\mathtt{SNR}_t$ . In other words, in this region, for a chosen $\mathtt{SNR}_t$, it is possible to obtain an optimal AP density in terms of received SINR CP (see the color plot in Fig.~\ref{Fig_color_contour}-{b}).

\begin{figure}[t]
 \begin{center}
 \subfigure[$\textmd{Contour plot}$]{\includegraphics[width = .45\textwidth]{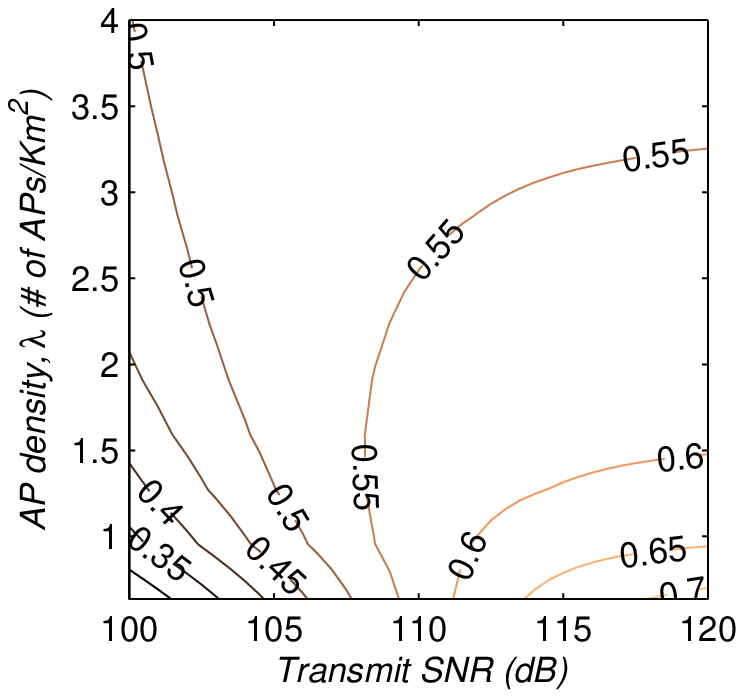}}
 \subfigure[$\textmd{Color plot}$]{\includegraphics[width = .45\textwidth]{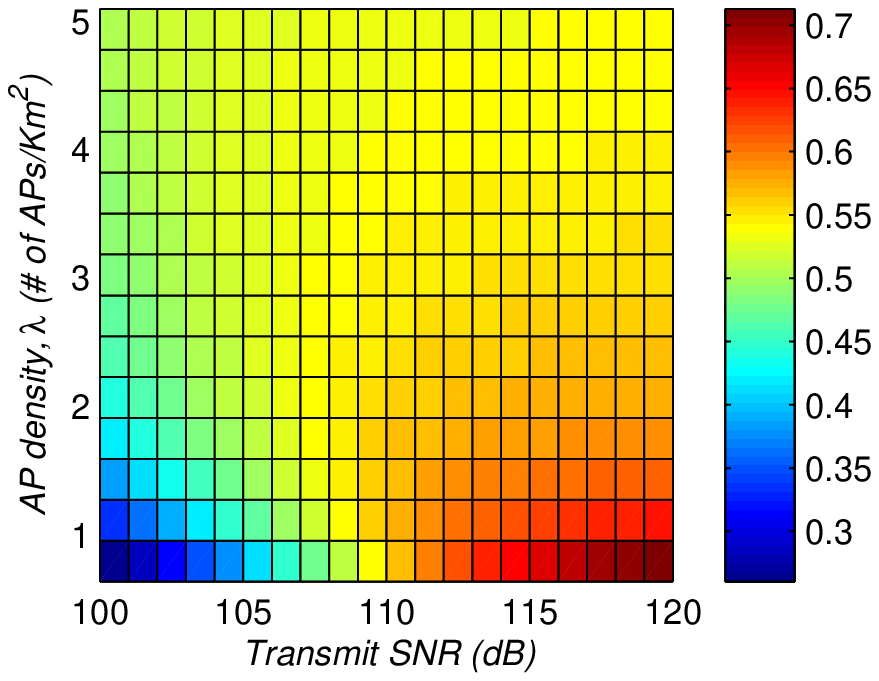}}
 \caption{The contour plot and the color plot of the worst SINR CP (for the target SINR of 0 dB) obtained based on different values of transmit SNR and AP density in a finite-area network with $\alpha =3.87$ and $\sigma_L = 6~\textmd{dB}$}
 \label{Fig_color_contour}
 \end{center}
\end{figure}

The relative behaviour of the highly-dense network as compared to a low-density network is further investigated in Fig.~\ref{Fig_7}. Fig.~\ref{Fig_7} illustrates the effect of the transmit SNR on the coverage probability at the center for the target SINR of $0~\textmd{dB}$ in a network with no shadowing. Again, the dotted lines represent the simulation results while the solid lines represent the analytical expression. As is clear from the figure, the infinite-area assumption (that matches the results in a dense network) underestimates low-density network performance for $\mathtt{SNR}_{\,t} \gtrsim 108~\textmd{dB}$ for the PLE of $\alpha = 3.87$. With a transmit signal power of $\sigma _{s}^{2} =20~\textmd{dBm}$, this corresponds to a noise variance of $\sigma _{n}^{2} \le -88{\kern 1pt}~\textmd{dBm}$ which, clearly, is common in practice. In a network with $\alpha= 3$, the transmit SNR threshold decreases to $\mathtt{SNR}_{\,t} \simeq 83~\textmd{dB}$. Nevertheless, in general, the range of transmit SNR for which the highly-dense network outperforms or falls behind a low-density network in SINR performance depends heavily on their relative AP densities and $\alpha$.

\begin{figure}[t]
 \begin{center}
 \includegraphics[width = .55\textwidth]{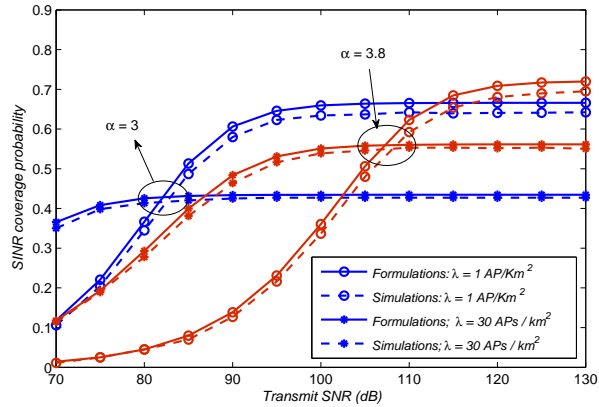}
 \caption{The effect of transmit SNR on the worst SINR coverage probability (for the target SINR of $0~\textmd{dB}$) in finite-area networks with AP densities $\lambda = 1~\textmd{AP/km}^2$ and $\lambda = 30~\textmd{APs/km}^2$ under no shadowing for the two PLEs of $\alpha = 3$ and $\alpha = 3.8$.}
 \label{Fig_7}
 \end{center}
 \end{figure}

Although the achievable SINR with respect to AP density depends heavily on the value of transmit SNR (decreases in an interference-limited network or increases in a noisy network with AP density), the pre-log factor $N$ in the user capacity formula ($C_{\,r_{1}} = N\log_{2}\left(1+\mathtt{SINR_{\,r_{1}}}\right)$) means that the user capacity increases monotonically with AP density irrespective of the noise variance. Figure~\ref{Fig_8} illustrates this effect on the worst (at the center) and maximum achievable user capacity coverage probability (for the target capacity of $C_{0}=5~\textmd{b/s/Hz}$)) within a circular finite-area network. As seen from Fig.~\ref{Fig_6} and Fig.~\ref{Fig_8}, even a highly dense finite-area network does not experience a uniform performance all through the region, rather, there is always a peak in the performance typically near the edges. As before, the infinite-area assumption does not
fully capture the behaviour of a finite-area network.

\begin{figure}[t]
\begin{center}
\subfigure[maximum user capacity coverage probability]{\includegraphics[width = .49\textwidth]{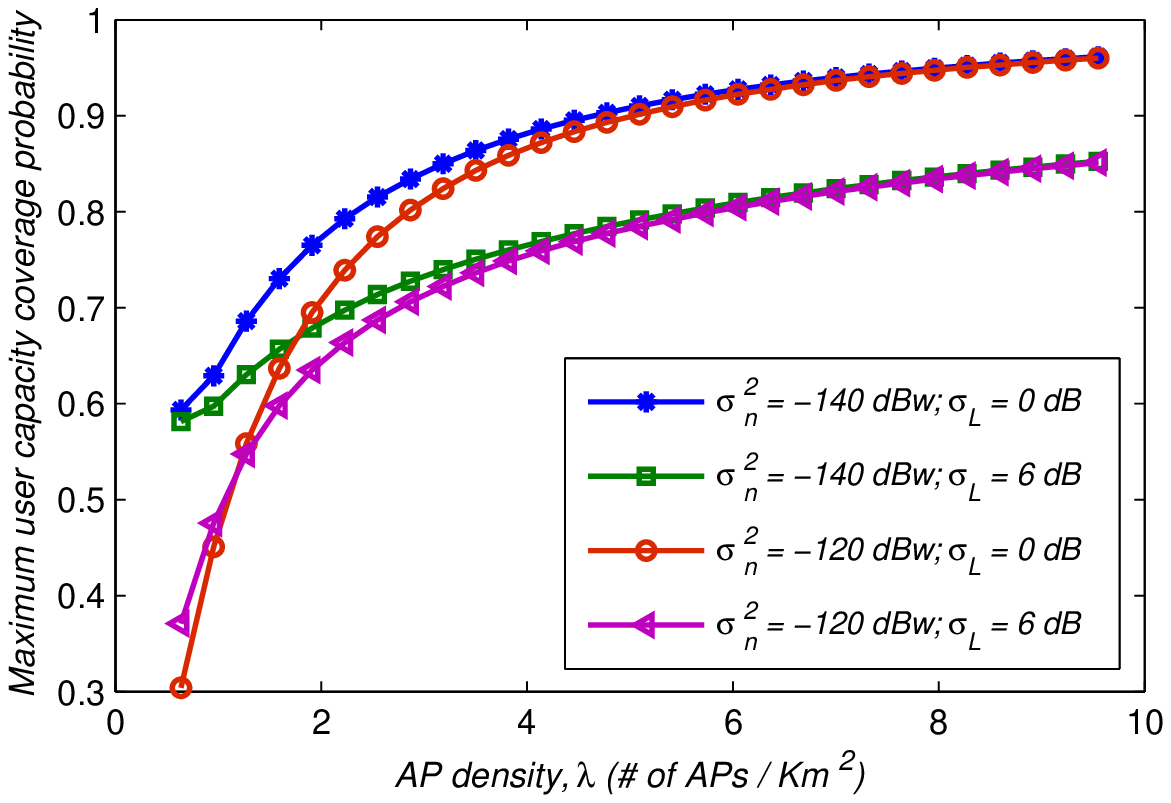}}
\subfigure[worst user capacity coverage probability]{\includegraphics[width = .49\textwidth]{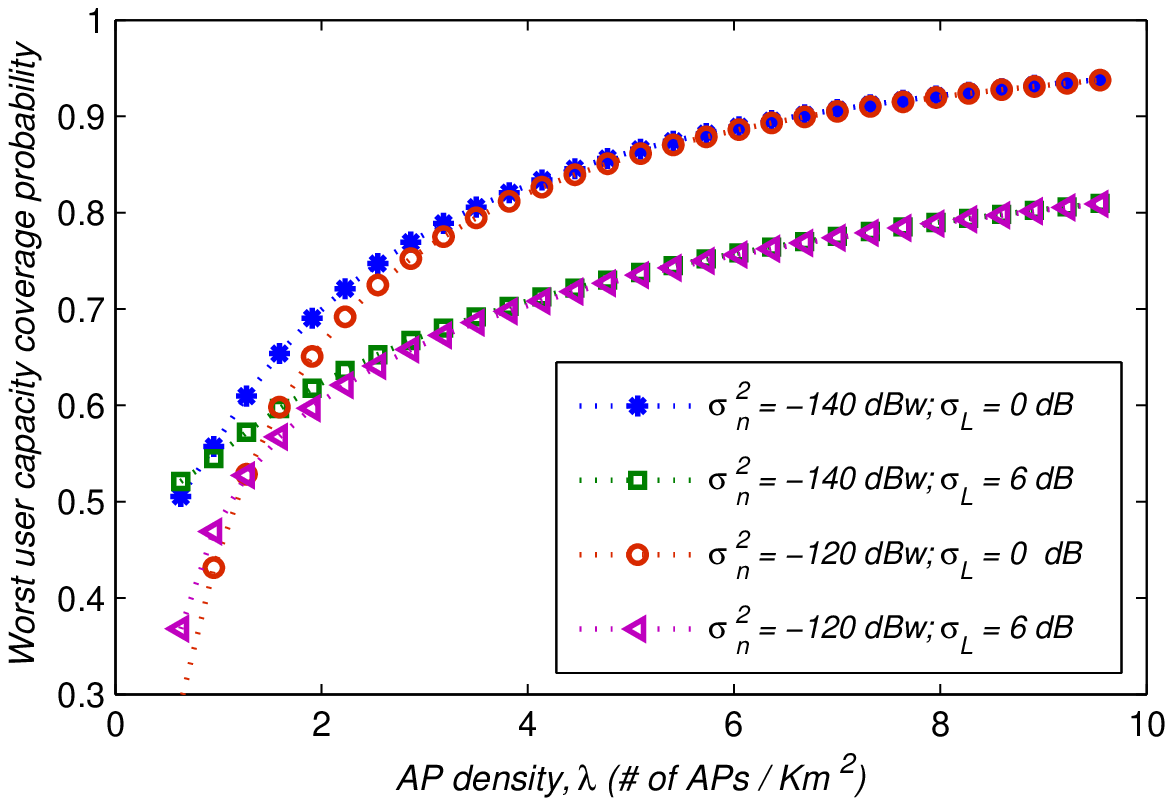}}
\caption{The effect of AP density on the a) maximum and b) worst achievable user capacity coverage probability (for the target capacity of $C_{0}=5~\textmd{b/s/Hz}$) in a finite-area network with PLEs of $\alpha =3.87$.}
\label{Fig_8}
\end{center}
\end{figure}
\begin{figure}[ht]
\begin{center}
 \includegraphics[width = .55\textwidth]{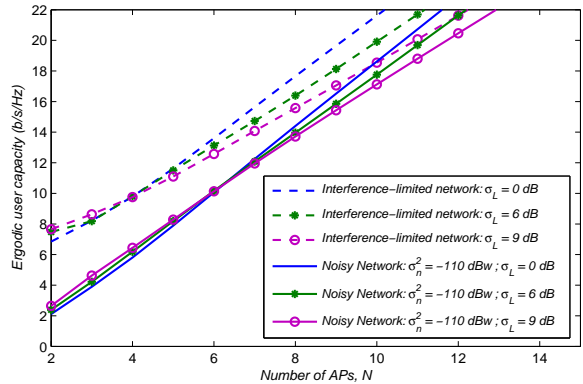}
\caption{The required number of APs versus target value of ergodic capacity for different values of $\sigma_{L}$ in a finite-area network with $\alpha =3.87$}
\label{Fig_9}
\end{center}
\end{figure}

\subsection{Design Example}

By using the worst-case user capacity at the centre of the circular network, we are able to answer the question as to how many APs are needed to guarantee a required target value of capacity in the network. For example, in the discussion associated with Fig.~\ref{Fig_8}-b, we chose a capacity coverage probability threshold of 0.6, i.e., we require that a user at any point in the network is able to achieve a capacity of $C_0=5 ~ \textmd{b/s/Hz}$ with probability 0.6. For such a requirement to be satisfied in an interference-limited finite-area network with $\sigma_L = 6~\textmd{dB}$, the network requires a minimum AP density of $\lambda = 1.59~\textmd{APs/Km}^2$. In a circular finite-area with radius $R_{\textmd{W}}= 1~\textmd{km}$, this corresponds to a minimum of $N = 5$ APs within the network. Under the same requirement, an interference-limited network with no shadowing requires one less AP as compared to the network with $\sigma_L = 6~\textmd{dB}$.

The design can be carried out for a target ergodic capacity as well. Figures~\ref{Fig_9} illustrates the relationship between the worst achievable ergodic user capacity and the number of APs $N$ for different values of $\sigma_L$ with $\alpha =3.87$. For a target value of capacity, Fig.~\ref{Fig_9} suggests a larger number of APs required for more severe shadowing environments. Moreover, the approximately linear relationship between the worst average user capacity and the number of APs is clear from the figures. This effect is expected beforehand from the model under consideration where the bandwidth allocated to each user grows with $N$.

Finally, Figure~\ref{Fig_ergodic} compares the approximate ergodic user capacity results obtained from the analytical formulation in~\eqref{eq_35_new} with those from the integration expression in~\eqref{eq_26} as well as the exact results from simulations for the two examples of $\sigma_L = 0~\textmd{dB}$ and $\sigma_L = 6~\textmd{dB}$ in an interference-limited network with $\alpha =4$. As is clear, there is a close match between the approximate and actual results (the approximate results are always within $10\%$ of the actual capacity for any $N$).
\begin{figure}[ht]
 \begin{center}
 \subfigure[$\sigma_L = 0~dB$]{\includegraphics[width = .48\textwidth]{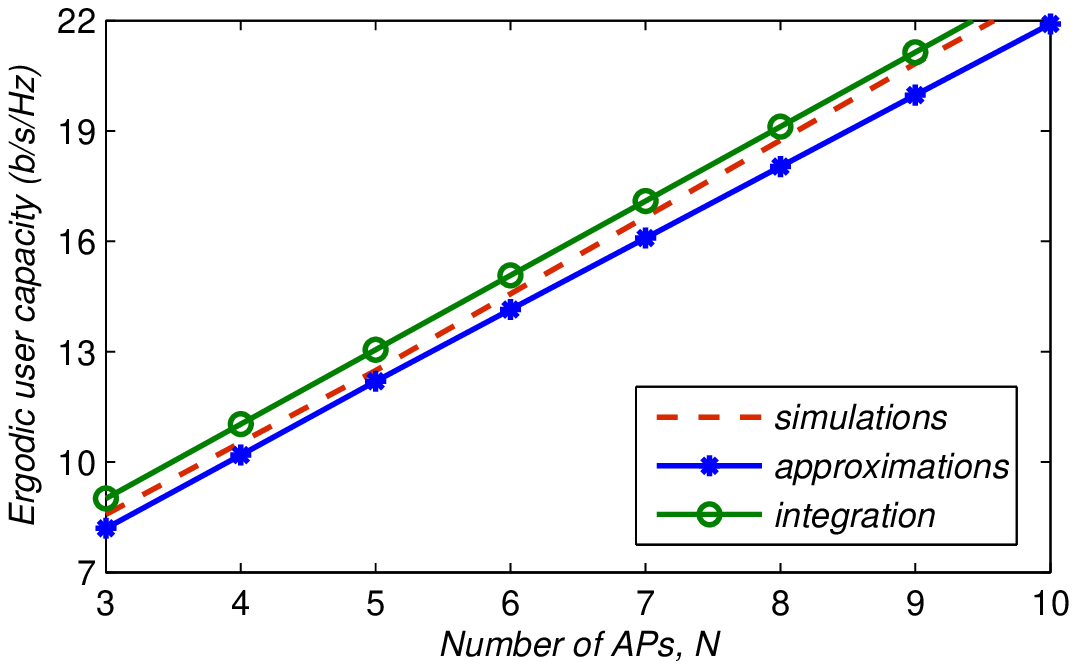}}
 \subfigure[$\sigma_L = 6~dB$]{\includegraphics[width = .48\textwidth]{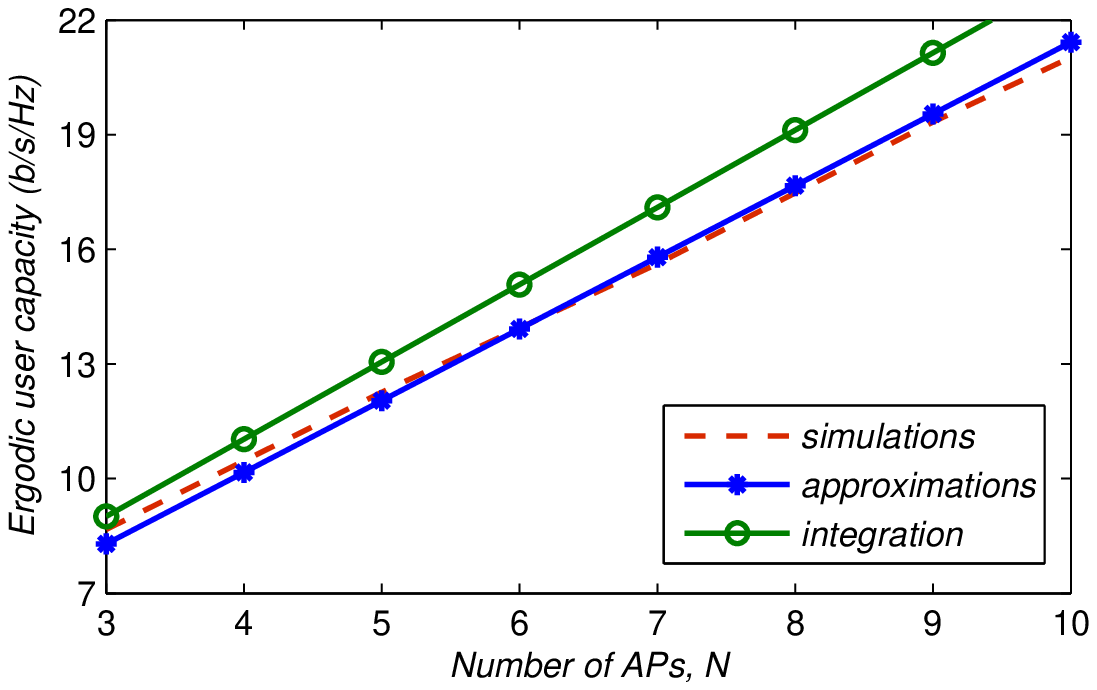}}
 \caption{Comparison between the approximate ergodic user capacity results obtained from the analytical formulation in~\eqref{eq_35_new} with those from the integration expression in~\eqref{eq_26} as well as the exact results from simulations in an interference-limited network with $\alpha =4$ for the two examples of: a) $\sigma_L = 0~\textmd{dB}$ ; b) $\sigma_L = 6~\textmd{dB}$}
 \label{Fig_ergodic}
 \end{center}
\end{figure}
It is worth noting that, Eq.~\eqref{eq_35_new} provides an approximation for a lower-bound on the integral~\eqref{eq_26}. Therefore, the fact that in Fig.~\ref{Fig_ergodic}-(b),~\eqref{eq_35_new} outperforms the numerical integration~\eqref{eq_26} when compared to the exact results, is purely by accident. For instance, as is seen from Fig.~\ref{Fig_ergodic}-(a), the above fact does not hold for $\sigma_L=0~\textmd{dB}$.

\section{Summary and Conclusions}\label{sec:Conc}

This paper analyzes the effect of AP density on the performance of a \emph{finite-area} network with a \emph{finite number} of uniformly distributed APs. Our motivation is two-fold - the available analyses in the literature are asymptotic and do not apply in the case of low-density networks and/or near edges of the finite area. As traditional cellular networks make way for newer network architectures, considering such a finite-area model is important for a better understanding of network capabilities and limitations. To further our analysis we obtain the achievable SINR coverage probability and the user capacity coverage probability at any point of the finite-area network. For practical values of "transmit SNR", $\mathtt{SNR}_{\,t} \ge 110~ \textmd{dB}$, the presented results are within 5\% of the actual results obtained from simulations.

The analysis also provides the specific loss in performance due to noise. In an interference-limited network, the SIR decreases monotonically with the number of APs, however, this is not the case when thermal noise is accounted for. In a finite-area network with a moderate noise variance, the SINR increases with $N$ (or AP density) and converges to the SINR of the interference-limited dense network. It has been reported earlier that an  infinite-area network underestimates the performance of a low-density interference-limited network. Correspondingly, the formulations allow a network designer to quantify the gain (loss) in performance from low values of AP density in an interference-limited network (and a noisy network) as compared to the highly-dense network. In particular, for the target SIR of $0~\textmd{dB}$ under no shadowing, an increase of at least 28\% in SIR coverage probability is obtained in an interference-limited network with $\alpha = 3.87$ and $\lambda = 1 \textmd{APs/km}^2$ as compared to a dense network with $\lambda = 30~\textmd{APs/km}^2$. The gain in SIR coverage probability increases to 56\% in an environment with $\alpha = 3$.

The formulation here accounts for different PLEs and network parameters, so they can lend themselves to parametric studies for network design. As an example of a parametric design, the worst-case user capacity coverage probability or average user capacity expression can be used to find the required number of APs (or AP density) to maintain the capacity at all points of the network above a target value.

\appendices
\section{Worst-Case Point}

The worst-case SINR, and hence capacity, is said to occur at the center of the circular finite-area network for small values of noise variance. In this appendix, we justify this claim. For a given point in $\textmd{W}$ with the associated distance $r_{\,1}$ to its nearest AP, the averaged SINR coverage probability is obtained as
\[\textmd{CP}_{\mathtt{SINR}}^{\textmd{\,avg}} (N, d, T, \alpha, \sigma_L) = \int_{0}^{R_{\textmd{W}} + d} \Px\{ \mathtt{SINR}_{\,r_{1}} > T \} f_{R_{\,1}}(r_{1})\,d\,r_{1}\]
where $\Px\{ \mathtt{SINR}_{\,r_{1}} > T \} \simeq Q\left( (\ln T - \mu _{\mathtt{SINR}})/\sigma _{\mathtt{SINR}} \right) $ and $f_{R_{\,1}}(r_{1})$ are the conditional SINR coverage probability (conditioned on $r_{\,1}$) and the PDF of the distance to the nearest AP, respectively. Due to the severe nonlinearities in $\Px\{ \mathtt{SINR}_{\,r_{1}} > T \}$ and $f_{R_{\,1}}(r_{1})$, the analytical proof for obtaining the location of the worst-case point is intractable. As an alternative, we resort to an intuitive explanation and simulation as to illustrate the worst-case average SINR occurs at the centre.

\begin{figure}[ht]
 \begin{center}
 \subfigure[]{\includegraphics[width = .49\textwidth]{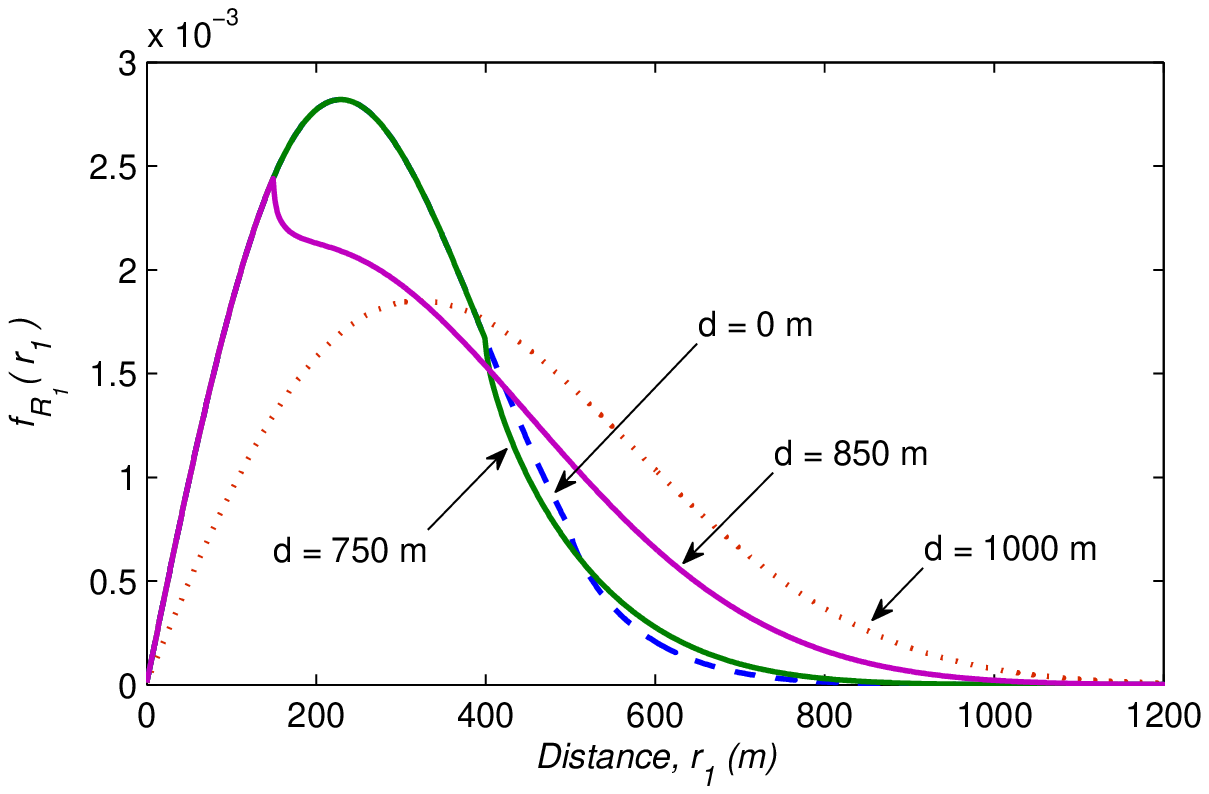}}
 \subfigure[]{\includegraphics[width = .49\textwidth]{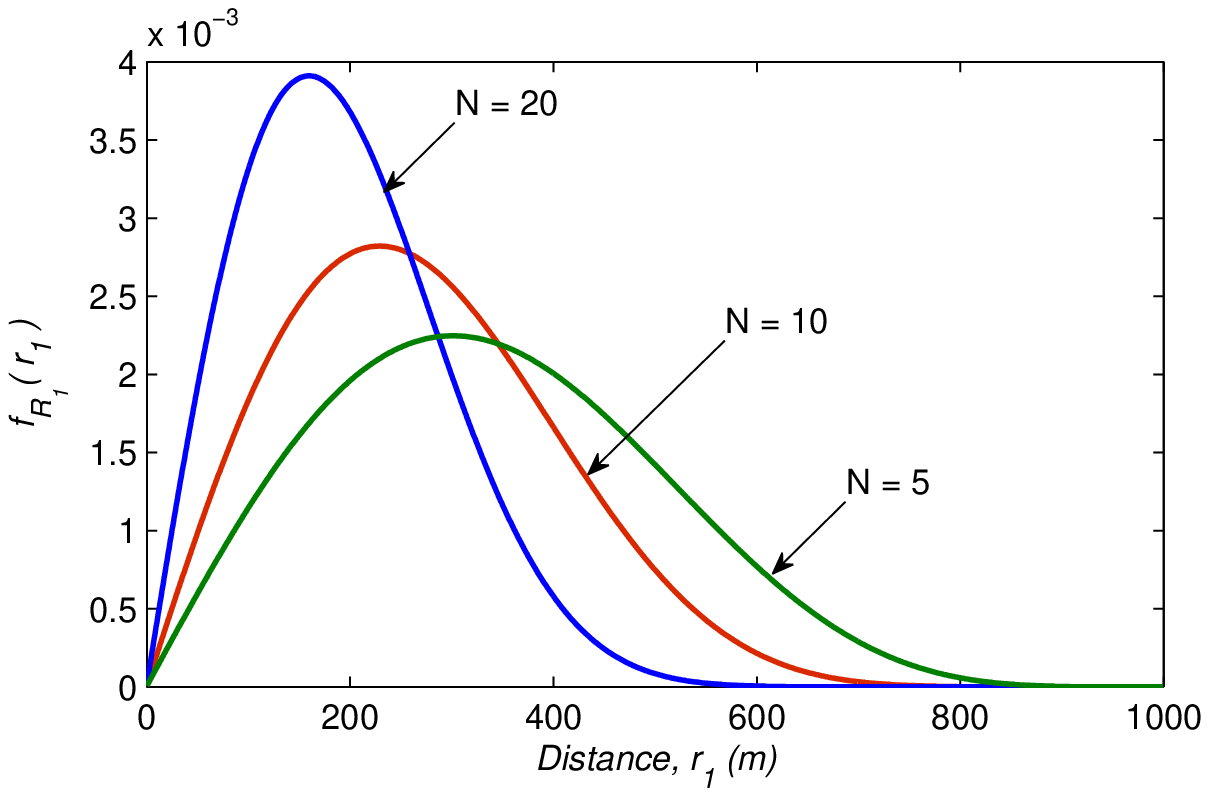}}
 \caption{Minimum distance PDF $f_{R_{\,1}}(r_{1})$ in a circular finite-area network with $R_{\textmd{W}} =1~ \textmd{km}$ for: a) different values of $d$ with $N = 10$; b) different values of $N$ with $d = 0$}
 \label{Appendix_I_Fig_1}
 \end{center}
\end{figure}

Figure~\ref{Appendix_I_Fig_1} illustrates the minimum distance PDF $f_{R_{\,1}}(r_{1})$ in a circular finite-area network with $R_{\textmd{W}} =1~ \textmd{km}$ for different values of $d$ and $N$. As is seen from Fig.~\ref{Appendix_I_Fig_1}-(a), for a given $N$, there is only a slight change in $f_{R_{\,1}}(r_{1})$ for $d \leq 0.75 R_{\textmd{W}}$. Therefore, we first concentrate on the area with $d \leq 0.75 R_{\textmd{W}}$.
Let $\textmd{Circ}(\textbf{x},R)$ denote a circle with the radius $R$ and the origin located at the point $\textbf{x}$. For a user located at the origin, denoted as $\textbf{o}$, the PDF of the distance to the nearest AP at $d = 0$, is given in~\eqref{eq_31} and is depicted in Fig.~\ref{Appendix_I_Fig_1}-(b) for different values of $N$ with $R_{\textmd{W}} =1~ \textmd{km}$. In this example, $f_{R_{1}}(r_{1})$ can be closely approximated by its truncated version for $0 \le r_{1} \le 0.75 R_{\textmd{W}} $ since $f_{R_{1}} (r_{1})$ is almost zero for $0.75 R_{\textmd{W}} \le r_{1} \le R_{\textmd{W}}$. This is like as if the largest possible distance to the nearest AP is $\bar{R} =0.75 R_{\textmd{W}}$ and $f_{R_{1}}(r_{1})$ is effectively non-zero only in $\textmd{Circ}(\textbf{o},0.75R_{\textmd{W}})$. The choice of $0.75 R_{\textmd{W}}$ is somewhat arbitrary and any reasonable choice would not change the justification.
\begin{figure}[t]
\begin{center}
\includegraphics[width = .35\textwidth]{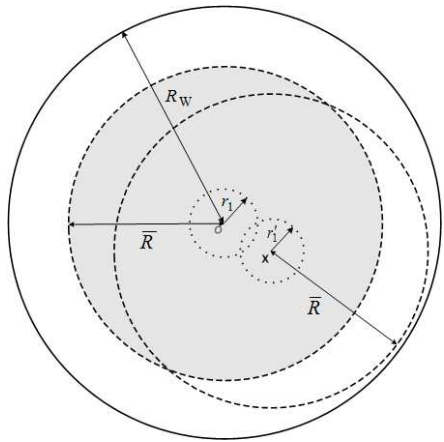}
\caption{Circular finite-area $\textmd{W}$ with radius $R_{\textmd{W}}$: the gray area indicates the region that the distance PDF associated with the user at the origin is effectively non-zero.}
\label{Appendix_I_Fig_2}
\end{center}
\end{figure}

Now consider a point $\textbf{x}$ within $\textmd{Circ}(\textbf{o},0.25 R_{\textmd{W}})$ as another user location within $\textmd{W}$ (see Fig.~\ref{Appendix_I_Fig_2}). Let $r\acute{}_{1}$ denotes the distance from point $\textbf{x}$ to the nearest AP. Since the $N$ APs are uniformly distributed in $\textmd{W}$, the PDF of $r\acute{}_{1}$, i.e., $f_{R_{1}} (r\acute{}_{1})$,  in the $\textmd{Circ}(\textbf{x},\bar{R} =0.75R_{\textmd{W}})$ would be the same as the truncated $f_{R_{1}}(r_{1})$ in $\textmd{Circ}(\textbf{o},0.75R_{\textmd{W}})$. As a result, for any $r_{1}=r\acute{}_{1}$, both users located at the two points of $\textbf{o}$ and $\textbf{x}$, receive the same signal power, on average, from their associated nearest APs. However, the two users do not experience the same interference power for $r_{1}=r\acute{}_{1}$. This is explained as follows. For the user at the origin $\textbf{o}$, the $N-1$ interfering APs are uniformly distributed between $\textmd{Circ}(\textbf{o}, R_{\textmd{W}})$ and $\textmd{Circ}(\textbf{o}, r_{1})$. On the other hand, the $N-1$ interfering APs for the user at $\textbf{x}$, are uniformly distributed in the area between $\textmd{Circ} (\textbf{o},R_{\textbf{W}})$ and $\textmd{Circ}(\textbf{x} , r\acute{}_{1}=r_{1})$, which is composed of  two  regions:  the   area   between  $\textmd{Circ}(\textbf{x},\bar{R})$ and $\textmd{Circ}(\textbf{x},r\acute{}_{1}=r_{1})$, and the area between $\textmd{Circ}(\textbf{o},R_{\textmd{W}})$ and $\textmd{Circ}(\textbf{x},\bar{R})$. The effect of the uniformly distributed interfering APs in the area between $\textmd{Circ}(\textbf{x},\bar{R})$ and $\textmd{Circ}(\textbf{x}, r\acute{}_{1}=r_{1})$ on the user at $\textbf{x}$, is the same as the one from uniformly distributed interfering APs in the area between $\textmd{Circ}(\textbf{o},\bar{R})$ and $\textmd{Circ}(\textbf{o},r_{1})$ on the user at $\textbf{o}$. However,  since  on  average,  the  distances  of the remaining interfering APs between
$\textmd{Circ}(\textbf{o} ,R_{\textmd{W}})$ and $\textmd{Circ}(\textbf{x},\bar{R})$ corresponding to the user at $\textbf{x}$, is larger than the distances of the remaining uniformly distributed interfering APs between $\textmd{Circ}(\textbf{o},R_{\textmd{W}})$ and $\textmd{Circ}(\textbf{o}, \bar{R})$ related to the user at $\textbf{o}$, the user at $\textbf{x}$ experiences lower interference power compared to the user at $\textbf{o}$. Thus, the user at the origin has the worst-case average SINR in the entire $\textmd{Circ}(\textbf{o} , 0.25R_{\textmd{W}})$.

For other points between $\textmd{Circ}(\textbf{o} , 0.25 R_{\textmd{W}})$ and $\textmd{Circ}(\textbf{o} , 0.75 R_{\textmd{W}})$, the PDF of distance to the nearest AP is nearly the same as $f_{R_{1}}(r_{1} \,|\,d = 0)$ (see Fig.~\ref{Appendix_I_Fig_1}-(a)). Therefore for a given distance $r_{1}$, the user receives the same signal power from its nearest AP as for the user at $\textbf{o}$, while it experiences much less interference power since the interfering APs are at a greater distance away, on average, compared to the ones for the user at $\textbf{o}$. As a result, it is intuitively clear that the worst-case SINR and so the worst-case user capacity is achieved at the centre of the circular area $\textmd{Circ}(\textbf{o} , 0.75 R_{\textmd{W}})$.

On the other hand, for larger values of $0.75 R_{\textmd{W}} \leq d \leq R_{\textmd{W}}$, the mean of random variable $r_{1}$ slightly increases with $d$. It follows that the average signal power decreases with $d$, resulting in the degradation of SINR in this region (this effect can be seen in Figs.~\ref{Fig_2D} -~\ref{Fig_3}.). However, still, since the user in this region experiences much less interference power in compared to the user at $\textbf{o}$, the worst-case SINR appears to be at the centre of $\textmd{Circ}(\textbf{o} , R_{\textmd{W}})$. Figs.~\ref{Fig_2D} - \ref{Fig_3} in the simulation section further justify this claim for different values of parameters of $N$, $\alpha$, and $\sigma_L$.

We note that in a noise-limited finite-area network, the above claim does not hold anymore. In this case the SNR remains approximately constant for $d \leq 0.75 R_{\textmd{W}}$  and degrades with $d$ for $d \geq 0.75 R_{\textmd{W}}$. It follows that the worst SNR occurs at the edge of the circular finite-area network.

\section{Approximation for the Worst Ergodic User Capacity under $\alpha=4$}

In Section~\ref{sec:SINR} the coverage probability and ergodic capacity were presented in terms of finite integrals that required numerical evaluation. In general, no closed form expressions are available; however, for the special case of the worst-case point (center) and for an integer choice of $\alpha$, some analysis is possible. Here, we present a closed-form approximation to the worst-case user capacity in an interference-limited network with $\alpha=4$~\cite{andrews2011tractable}. The worst achievable user capacity (in b/s/Hz) averaged over different realizations of AP locations is given by
\begin{equation} \label{eq_36_new}
C_{\textmd{\,ergodic}}^{\textmd{\,worst}} = \int_{0}^{R_{\textmd{W}}} C_{\textmd{\,ergodic} | r_{1}}^{\textmd{\,worst}} f_{R_{1}}(r_{1}) \textmd{d} r_{1}  =\int _{0}^{R_{\textmd{W}}} \Ex \left\{ N \log_{2} (1+\mathtt{SIR}_{\,r_{1}}) | r_{1} \right\} f_{R_{1}} (r_{1}) \textmd{d}r_{1},
\end{equation}
where $f_{R_{1}}(r_{1})$ is given in~\eqref{eq_31} and the ergodic capacity for a given $r_{1} $ is the ensemble average over different realizations of the channels in~\eqref{eq_2}.

For a given $r_{1} $, the ergodic capacity $C_{\textmd{\,ergodic}\,| \, r_{1} }^{\textmd{\,worst}} =\Ex\left\{ N \log_{2} (1+\mathtt{SIR}_{\,r_{1}}) | r_{1} \right\}$ is upper-bounded by $N \log _{2} (1+\Ex\{ \mathtt{SIR}_{\,r_{1}} |r_{1} \} )$. On the other hand, since the mean $\Ex\{ \mathtt{SIR}_{\,r_{1} } |r_{1} \} =\exp{(\mu_{\,\textmd{SIR}} +\sigma _{\textmd{\,SIR}}^{2} /2)}$ is always greater than one, the $\log_{2} (1+\Ex\{ \mathtt{SIR}_{\,r_{1} } |\,r_{1} \})$ itself is lower-bounded by $\log_{2} (\Ex\{ \mathtt{SIR}_{\,r_{1}} |\,r_{1} \})$. In general, there is no guarantee that $N \log_{2} (\Ex\{ \mathtt{SIR}_{\,r_{1} } |\,r_{1} \})$ is lower than $C_{\,\textmd{ergodic} \, | \, r_{1} }^{\,\textmd{worst}}$. However, for small values of $\mu _{\textmd{SIR}}$ and $\sigma_{\textmd{SIR}}$, $N \log_{2} (\Ex\{\mathtt{SIR}_{\,r_{1} } |r_{1} \} )\le C_{\,\textmd{ergodic}\, |\, r_{1} }^{\,\textmd{worst}}$ holds. Therefore, an approximate lower-bound on the $C_{\,\textmd{ergodic} \, |\, r_{1} }^{\,\textmd{worst}}$, can be obtained as
\begin{equation} \label{eq_37_new}
C_{\textmd{ergodic} \, | \, r_{1} }^{\,\textmd{worst}} \approx N \log_{2} (\Ex\{ \mathtt{SIR}_{\,r_{1}} |\,r_{1} \} )=N(\mu_{\textmd{SIR}} +\sigma_{\textmd{SIR}}^{2} /2)/\ln 2.
\end{equation}

In an interference-limited network with $\alpha =4$, $\mu_{\textmd{Denom}}$ and $\sigma_{\textmd{Denom}}^{2}$ from~\eqref{eq_mu_I}-\eqref{eq_sigma_I} simplify to
\begin{equation} \label{eq_38_new}
\begin{array}{l} {\mu_{\textmd{Denom}} =\ln \bigg[\frac{(N-1)^{2} \sigma _{s}^{4} e^{\sigma _{z}^{2} } \left(\frac{r_{1}^{-2} -R_{\textmd{W}}^{-2} }{R_{\textmd{W}}^{2} -r_{1}^{2} } \right)^{2} }{\big((N-1)(N-2)\sigma _{s}^{4} e^{\sigma _{z}^{2} } \left(\frac{r_{1}^{-2} -R_{\textmd{W}}^{-2} }{R_{\textmd{W}}^{2} -r_{1}^{2} } \right)^{2} \big)^{1/2} \big(1+\frac{2}{3(N-2)} (\underbrace{1+ r_{1}^{2} /R_{\textmd{W}}^{2} +R_{\textmd{W}}^{2} /r_{1}^{2}}_{\approx R_{\textmd{W}}^{2} /r_{1}^{2}}) \big)^{1/2} } \bigg]} \\ {{\kern 1pt} {\kern 1pt} {\kern 1pt} {\kern 1pt} {\kern 1pt} {\kern 1pt} {\kern 1pt} {\kern 1pt} {\kern 1pt} {\kern 1pt} {\kern 1pt} {\kern 1pt} {\kern 1pt} {\kern 1pt} {\kern 1pt} {\kern 1pt} {\kern 1pt} {\kern 1pt} {\kern 1pt} {\kern 1pt} {\kern 1pt} {\kern 1pt} {\kern 1pt} {\kern 1pt} {\kern 1pt} {\kern 1pt} {\kern 1pt} {\kern 1pt} {\kern 1pt} {\kern 1pt} {\kern 1pt} {\kern 1pt} \approx \ln \left[\frac{(N-1)^{3/2} \sigma _{s}^{2} e^{\sigma _{z}^{2} /2} }{(N-2)^{1/2} } \right]+\ln \left[\frac{r_{1}^{-2} -R_{\textmd{W}}^{-2} }{R_{\textmd{W}}^{2} -r_{1}^{2} } \right]-\frac{1}{2} \ln \left[1+\frac{2e^{\sigma _{z}^{2} } R_{\textmd{W}}^{2} }{3(N-2){\kern 1pt} {\kern 1pt} r_{1}^{2} } \right]} \end{array},
\end{equation}
\begin{equation} \label{eq_39_new}
\begin{array}{l} {\sigma_{\textmd{Denom}}^{2} = \ln \bigg[\frac{N-2}{N-1} +\frac{2e^{\sigma _{z}^{2} } }{3(N-1)} (\underbrace{1+\frac{r_{1}^{2} }{R_{\textmd{W}}^{2} } +\frac{R_{\textmd{W}}^{2} }{r_{1}^{2} } }_{\approx R_{\textmd{W}}^{2} /r_{1}^{2} } )\bigg]\approx \ln \left[\left(\frac{N-2}{N-1} \right)\left(1+\frac{2e^{\sigma _{z}^{2} } R_{\textmd{W}}^{2} }{3(N-2){\kern 1pt} r_{1}^{2} } \right)\right]}. \end{array}
\end{equation}

\noindent In the first line of \eqref{eq_38_new}-\eqref{eq_39_new}, the term $(1+r_{1}^{2} /R_{\textmd{W}}^{\,2} +R_{\textmd{W}}^{\,2} /r_{1}^{2})$ is approximated by $R_{\textmd{W}}^{\,2} /r_{1}^{2} $. This is justified as follows. As is seen from Fig.~\ref{Appendix_I_Fig_1}-(b), $f_{R_{1}}(r_{1})$ at the worst-case point is small for $0.75 R_{\textmd{W}} \le r_{1} \le {\kern 1pt} R_{\textmd{W}} $. As a result, for the range of values of $r_{1} $ that effectively contribute in the average capacity~\eqref{eq_36_new}, $R_{\textmd{W}}^{\,2} /r_{1}^{2} \gg 1$ and so the approximation is valid. Following~\eqref{eq_38_new}-\eqref{eq_39_new},
\begin{equation} \label{eq_40_new}
\mu_{\textmd{SIR}} \approx 2\ln R_{\textmd{W}} -2\ln r_{1} +\ln \frac{(N-2)^{1/2} }{(N-1)^{3/2}} -\ln \sqrt{2} -\sigma _{z}^{2} /2+\frac{1}{2} \ln \bigg[1+\frac{2e^{\sigma _{z}^{2} } R_{\textmd{W}}^{2} }{3(N-2){\kern 1pt} r_{1}^{2} } \bigg],
\end{equation}
\begin{equation} \label{eq_41_new}
\sigma_{\textmd{SIR}}^{2} \approx \sigma_{z}^{2} +\ln 2+\ln \frac{N-2}{N-1} +\ln \bigg[1+\frac{2e^{\sigma _{z}^{2} } R_{\textmd{W}}^{2} }{3(N-2){\kern 1pt} r_{1}^{2} } \bigg],
\end{equation}
which gives the conditional ergodic capacity $C_{\textmd{ergodic} \, | \, r_{1} }^{\,\textmd{worst}} $ as
\begin{equation} \label{eq_42_new}
C_{\textmd{ergodic}\,  |\, r_{1} }^{\,\textmd{worst}} =\frac{N}{\ln 2} \bigg( 2\ln R_{\textmd{W}} -2\ln r_{1} +\ln \frac{(N-2)^{1/2} }{(N-1)^{3/2} } +\frac{1}{2} \ln \frac{N-2}{N-1} +\ln \bigg[1+\frac{2e^{\sigma _{z}^{2} } R_{\textmd{W}}^{2} }{3(N-2){\kern 1pt} r_{1}^{2} } \bigg] \bigg).
\end{equation}

Now the worst ergodic user capacity in~\eqref{eq_36_new} is given as
\begin{equation} \label{eq_43_new}
\begin{array}{l} {C_{\,\textmd{ergodic}}^{\,\textmd{worst}} = \displaystyle\int_{0}^{R_{\textmd{W}} }C_{\,\textmd{ergodic} | r_{1} }^{\,\textmd{worst}} f_{R_{1}}(r_{1}) \textmd{d}\,r_{1}  =\frac{N}{\ln 2} \left(2\ln R_{\textmd{W}} +\ln \frac{(N-2)^{1/2} }{(N-1)^{3/2} } +\frac{1}{2} \ln \frac{N-2}{N-1} \right)} \\ {{\kern 1pt} {\kern 1pt} {\kern 1pt} {\kern 1pt} {\kern 1pt} {\kern 1pt} {\kern 1pt} {\kern 1pt} {\kern 1pt} {\kern 1pt} {\kern 1pt} {\kern 1pt} {\kern 1pt} {\kern 1pt} {\kern 1pt} {\kern 1pt} {\kern 1pt} {\kern 1pt} {\kern 1pt} {\kern 1pt} {\kern 1pt} {\kern 1pt} {\kern 1pt} {\kern 1pt} {\kern 1pt} {\kern 1pt} {\kern 1pt} {\kern 1pt} {\kern 1pt} {\kern 1pt} {\kern 1pt} {\kern 1pt} {\kern 1pt} {\kern 1pt} {\kern 1pt} {\kern 1pt} {\kern 1pt} {\kern 1pt} {\kern 1pt} {\kern 1pt} {\kern 1pt} {\kern 1pt} {\kern 1pt} {\kern 1pt} {\kern 1pt} {\kern 1pt} {\kern 1pt} {\kern 1pt} {\kern 1pt} {\kern 1pt} {\kern 1pt} {\kern 1pt} {\kern 1pt} {\kern 1pt} {\kern 1pt} {\kern 1pt} +\frac{N}{\ln 2} \bigg(\underbrace{\int _{0}^{R_{\textmd{W}} }-2\ln r_{1} f_{R_{1}}(r_{1})\textmd{d}r_{1}}_{\left(*\right)} +\underbrace{\int _{0}^{R_{\textmd{W}} }\ln \bigg[1+\frac{2e^{\sigma _{z}^{2} } R_{\textmd{W}}^{2} }{3(N-2){\kern 1pt} r_{1}^{2} } \bigg]f_{R_{1}} (r_{1} )\textmd{d}r_{1}}_{\left(**\right)} \bigg)} \end{array}
\end{equation}

All that remains is to solve the integrals in $(*)$ and $(**)$. Using the binomial equivalence for $(1+x)^{n} =\sum_{k=0}^{n} \binom{n}{k} x^{k} $, the distance PDF $f_{R_{\,1}}(r_{1})$ in \eqref{eq_31} can be rewritten as

\begin{equation} \label{eq_44_new}
f_{R_{\,1}}(r_{1} )=\frac{2Nr_{1}}{R_{\textmd{W}}^{2}} \sum_{k=0}^{N-1}(-1)^{k} \binom{N-1}{k} \frac{r_{1}^{2k} }{R_{\textmd{W}}^{\,2k} }  {\kern 1pt} {\kern 1pt} {\kern 1pt} {\kern 1pt} ;{\kern 1pt} {\kern 1pt} {\kern 1pt} {\kern 1pt} 0\le r_{1} \le R_{\textmd{W}} .
\end{equation}

Now, the integration in $(*)$ follows as
\begin{equation} \label{eq_45_new}
\begin{array}{l} { \int_{0}^{R_{\textmd{W}} }-2\ln r_{1} f_{R_{\,1}}(r_{1} )\textmd{d}r_{1} = -\frac{4N}{R_{\textmd{W}}^{\,2} } \sum\limits_{k=0}^{N-1} \left[\frac{(-1)^{k} }{R_{\textmd{W}}^{\,2k} } \binom{N-1}{k} \int _{0}^{R_{\textmd{W}} }r_{1}^{2k+1} \ln r_{1} {\kern 1pt} \textmd{d} r_{1}  \right]} \\ {{\kern 1pt} {\kern 1pt} {\kern 1pt} {\kern 1pt} {\kern 1pt} {\kern 1pt} {\kern 1pt} {\kern 1pt} {\kern 1pt} {\kern 1pt} {\kern 1pt} {\kern 1pt} {\kern 1pt} {\kern 1pt} {\kern 1pt} {\kern 1pt} {\kern 1pt} {\kern 1pt} {\kern 1pt} {\kern 1pt} {\kern 1pt} {\kern 1pt} {\kern 1pt} {\kern 1pt} {\kern 1pt} {\kern 1pt} {\kern 1pt} {\kern 1pt} {\kern 1pt} {\kern 1pt} {\kern 1pt} {\kern 1pt} {\kern 1pt} {\kern 1pt} {\kern 1pt} {\kern 1pt} {\kern 1pt} {\kern 1pt} {\kern 1pt} {\kern 1pt} {\kern 1pt} {\kern 1pt} {\kern 1pt} {\kern 1pt} {\kern 1pt} {\kern 1pt} {\kern 1pt} {\kern 1pt} {\kern 1pt} {\kern 1pt} {\kern 1pt} =-4N\ln R_{\textmd{W}} \sum\limits_{k=0}^{N-1}(-1)^{k} \binom{N-1}{k} \frac{1}{2k+2}  +4N\sum\limits_{k=0}^{N-1}(-1)^{k} \binom{N-1}{k} \frac{1}{(2k+2)^{2}}} \\ {{\kern 1pt} {\kern 1pt} {\kern 1pt} {\kern 1pt} {\kern 1pt} {\kern 1pt} {\kern 1pt} {\kern 1pt} {\kern 1pt} {\kern 1pt} {\kern 1pt} {\kern 1pt} {\kern 1pt} {\kern 1pt} {\kern 1pt} {\kern 1pt} {\kern 1pt} {\kern 1pt} {\kern 1pt} {\kern 1pt} {\kern 1pt} {\kern 1pt} {\kern 1pt} {\kern 1pt} {\kern 1pt} {\kern 1pt} {\kern 1pt} {\kern 1pt} {\kern 1pt} {\kern 1pt} {\kern 1pt} {\kern 1pt} {\kern 1pt} {\kern 1pt} {\kern 1pt} {\kern 1pt} {\kern 1pt} {\kern 1pt} {\kern 1pt} {\kern 1pt} {\kern 1pt} {\kern 1pt} {\kern 1pt} {\kern 1pt} {\kern 1pt} {\kern 1pt} {\kern 1pt} {\kern 1pt} {\kern 1pt} {\kern 1pt} {\kern 1pt} =-2N\ln R_{\textmd{W}} \sum\limits_{k=0}^{N-1}\frac{(-1)^{k} }{N} \binom{N-1}{k+1} +N\sum\limits_{k=0}^{N-1} \frac{(-1)^{k}}{N} \binom{N-1}{k+1} \frac{1}{k+1}} \\ {{\kern 1pt} {\kern 1pt} {\kern 1pt} {\kern 1pt} {\kern 1pt} {\kern 1pt} {\kern 1pt} {\kern 1pt} {\kern 1pt} {\kern 1pt} {\kern 1pt} {\kern 1pt} {\kern 1pt} {\kern 1pt} {\kern 1pt} {\kern 1pt} {\kern 1pt} {\kern 1pt} {\kern 1pt} {\kern 1pt} {\kern 1pt} {\kern 1pt} {\kern 1pt} {\kern 1pt} {\kern 1pt} {\kern 1pt} {\kern 1pt} {\kern 1pt} {\kern 1pt} {\kern 1pt} {\kern 1pt} {\kern 1pt} {\kern 1pt} {\kern 1pt} {\kern 1pt} {\kern 1pt} {\kern 1pt} {\kern 1pt} {\kern 1pt} {\kern 1pt} {\kern 1pt} {\kern 1pt} {\kern 1pt} {\kern 1pt} {\kern 1pt} {\kern 1pt} {\kern 1pt} {\kern 1pt} {\kern 1pt} {\kern 1pt} {\kern 1pt} =2\ln R_{\textmd{W}} \bigg[\sum\limits_{k'=0}^{N}(-1)^{k'} \binom{N}{k'} -1 \bigg]+\sum\limits_{k'=1}^{N}(-1)^{k'-1} \binom{N}{k'} \int _{0}^{1}u^{k'-1} \textmd{d}u} \\ {{\kern 1pt} {\kern 1pt} {\kern 1pt} {\kern 1pt} {\kern 1pt} {\kern 1pt} {\kern 1pt} {\kern 1pt} {\kern 1pt} {\kern 1pt} {\kern 1pt} {\kern 1pt} {\kern 1pt} {\kern 1pt} {\kern 1pt} {\kern 1pt} {\kern 1pt} {\kern 1pt} {\kern 1pt} {\kern 1pt} {\kern 1pt} {\kern 1pt} {\kern 1pt} {\kern 1pt} {\kern 1pt} {\kern 1pt} {\kern 1pt} {\kern 1pt} {\kern 1pt} {\kern 1pt} {\kern 1pt} {\kern 1pt} {\kern 1pt} {\kern 1pt} {\kern 1pt} {\kern 1pt} {\kern 1pt} {\kern 1pt} {\kern 1pt} {\kern 1pt} {\kern 1pt} {\kern 1pt} {\kern 1pt} {\kern 1pt} {\kern 1pt} {\kern 1pt} {\kern 1pt} {\kern 1pt} {\kern 1pt} {\kern 1pt} {\kern 1pt} =2 \ln R_{\textmd{W}} \big[(1-x)^{N} \left|_{x=1} \right. -1\big]+ \int_{0}^{1} \bigg[\sum\limits_{k'=1}^{N}(-1)^{k'-1} \binom{N}{k'} u^{k'-1}  \bigg] \textmd{d}u} \\ {{\kern 1pt} {\kern 1pt} {\kern 1pt} {\kern 1pt} {\kern 1pt} {\kern 1pt} {\kern 1pt} {\kern 1pt} {\kern 1pt} {\kern 1pt} {\kern 1pt} {\kern 1pt} {\kern 1pt} {\kern 1pt} {\kern 1pt} {\kern 1pt} {\kern 1pt} {\kern 1pt} {\kern 1pt} {\kern 1pt} {\kern 1pt} {\kern 1pt} {\kern 1pt} {\kern 1pt} {\kern 1pt} {\kern 1pt} {\kern 1pt} {\kern 1pt} {\kern 1pt} {\kern 1pt} {\kern 1pt} {\kern 1pt} {\kern 1pt} {\kern 1pt} {\kern 1pt} {\kern 1pt} {\kern 1pt} {\kern 1pt} {\kern 1pt} {\kern 1pt} {\kern 1pt} {\kern 1pt} {\kern 1pt} {\kern 1pt} {\kern 1pt} {\kern 1pt} {\kern 1pt} {\kern 1pt} {\kern 1pt} {\kern 1pt} {\kern 1pt} =-2\ln R_{\textmd{W}} + \int_{0}^{1}\frac{1}{u} \big[1-(1-u)^{N} \big] \textmd{d}u} \\
{{\kern 1pt} {\kern 1pt} {\kern 1pt} {\kern 1pt} {\kern 1pt} {\kern 1pt} {\kern 1pt} {\kern 1pt} {\kern 1pt} {\kern 1pt} {\kern 1pt} {\kern 1pt} {\kern 1pt} {\kern 1pt} {\kern 1pt} {\kern 1pt} {\kern 1pt} {\kern 1pt} {\kern 1pt} {\kern 1pt} {\kern 1pt} {\kern 1pt} {\kern 1pt} {\kern 1pt} {\kern 1pt} {\kern 1pt} {\kern 1pt} {\kern 1pt} {\kern 1pt} {\kern 1pt} {\kern 1pt} {\kern 1pt} {\kern 1pt} {\kern 1pt} {\kern 1pt} {\kern 1pt} {\kern 1pt} {\kern 1pt} {\kern 1pt} {\kern 1pt} {\kern 1pt} {\kern 1pt} {\kern 1pt} {\kern 1pt} {\kern 1pt} {\kern 1pt} {\kern 1pt} {\kern 1pt} {\kern 1pt} {\kern 1pt} {\kern 1pt} =-2\ln R_{\textmd{W}} + \int_{{\kern 1pt} 0}^{1}\frac{1-x^{N} }{1-x} \textmd{d}x =-2\ln R_{\textmd{W}} +H_{N}} \end{array}
\end{equation}
In the last line, $\int_{0}^{1}(1-x^{N} )/(1-x)\textmd{d}x $ is an integral representation of the \textit{N}-th harmonic number, $H_{N} =\sum _{k=1}^{N}1/k $, given by Euler. The corresponding expansion of $H_{N} $ is given as~\cite{havil2003exploring}

\begin{equation} \label{eq_46_new}
H_{N} \approx \ln N+\gamma +\frac{1}{2N} -\sum _{k=1}^{\infty }\frac{\xi _{2k} }{2kN^{2k} }  ,
\end{equation}
where $\gamma \approx 0.578$ is the Euler-Mascheroni constant and $\xi _{k};k=1,\cdots$ are the Bernoulli numbers. Approximating $H_{N}$ with the first three terms in \eqref{eq_46_new}, the integration in \eqref{eq_45_new} is given as
\begin{equation} \label{eq_47_new}
\int_{0}^{R_{\textmd{W}}}-2\ln r_{1} f_{R_{\,1} }(r_{1})\textmd{d} r_{1}  \approx -2\ln R_{\textmd{W}} +\ln N+\gamma +\frac{1}{2N} .
\end{equation}

For the integration denoted as $(**)$, let $b=2e^{\sigma _{z}^{2}} R_{\textmd{W}}^{2} /(3(N-2))$. The integration follows as
\begin{equation} \label{eq_48_new}
\begin{array}{l} {\displaystyle\int_{0}^{R_{\textmd{W}} }\ln \big[1+\frac{b}{{\kern 1pt} r_{1}^{2} } \big] f_{R_{\,1}}(r_{1} )\textmd{d} r_{1} =\int _{0}^{R_{\textmd{W}} }-2\ln r_{1} f_{R_{\,1} } (r_{1} )\textmd{d} r_{1}  + \int _{0}^{R_{\textmd{W}}}\ln (r_{1}^{2} +b)f_{R_{\,1}}(r_{1}) \textmd{d} r_{1}} \\ {{\kern 1pt} {\kern 1pt} {\kern 1pt} {\kern 1pt} {\kern 1pt} {\kern 1pt} {\kern 1pt} {\kern 1pt} {\kern 1pt} {\kern 1pt} {\kern 1pt} {\kern 1pt} {\kern 1pt} {\kern 1pt} {\kern 1pt} {\kern 1pt} {\kern 1pt} {\kern 1pt} {\kern 1pt} {\kern 1pt} {\kern 1pt} {\kern 1pt} {\kern 1pt} {\kern 1pt} {\kern 1pt} {\kern 1pt} {\kern 1pt} {\kern 1pt} {\kern 1pt} {\kern 1pt} {\kern 1pt} {\kern 1pt} {\kern 1pt} {\kern 1pt} {\kern 1pt} {\kern 1pt} {\kern 1pt} {\kern 1pt} {\kern 1pt} {\kern 1pt} {\kern 1pt} {\kern 1pt} {\kern 1pt} {\kern 1pt} {\kern 1pt} {\kern 1pt} {\kern 1pt} {\kern 1pt} {\kern 1pt} {\kern 1pt} {\kern 1pt} {\kern 1pt} {\kern 1pt} {\kern 1pt} {\kern 1pt} {\kern 1pt} {\kern 1pt} {\kern 1pt} {\kern 1pt} {\kern 1pt} {\kern 1pt} {\kern 1pt} {\kern 1pt} {\kern 1pt} {\kern 1pt} {\kern 1pt} {\kern 1pt} {\kern 1pt} {\kern 1pt} {\kern 1pt} {\kern 1pt} {\kern 1pt} {\kern 1pt} {\kern 1pt} {\kern 1pt} {\kern 1pt} {\kern 1pt} {\kern 1pt} {\kern 1pt} {\kern 1pt} {\kern 1pt} {\kern 1pt} {\kern 1pt} {\kern 1pt} {\kern 1pt} {\kern 1pt} {\kern 1pt} {\kern 1pt} {\kern 1pt} {\kern 1pt} {\kern 1pt} {\kern 1pt} {\kern 1pt} {\kern 1pt} {\kern 1pt} {\kern 1pt} {\kern 1pt} {\kern 1pt} {\kern 1pt} {\kern 1pt} \approx -2\ln R_{\textmd{W}} +\ln N+\gamma +\frac{1}{2N} +\underbrace{\int_{0}^{R_{\textmd{W}} }\ln (r_{1}^{2} +b)f_{r_{1} } (r_{1} )d_{r_{1} }  }_{\left(***\right)} } \end{array}.
\end{equation}

Using integration by parts and letting $u=\ln [\,r_{1}^{2} +b\,]$ and $dv=r_{1}^{2k+1} dr_{1} $, the integration $(***)$ can be rewritten as
\begin{equation} \label{eq_49_new}
\begin{array}{l} {(***) = \displaystyle\int_{0}^{R_{\textmd{W}}}\ln [\,r_{1}^{2} +b \,]\frac{2Nr_{1} }{R_{\textmd{W}}^{2} } \sum\limits_{k=0}^{N-1}(-1)^{k} \binom{N-1}{k} \frac{r_{1}^{2k} }{R_{\textmd{W}}^{\,2k}} \textmd{d} r_{1}} \\ {~~~~~~~~~ =\frac{2N}{R_{\textmd{W}}^{\,2} } \sum\limits_{k=0}^{N-1}\frac{(-1)^{k} }{R_{\textmd{W}}^{\,2k}} \binom{N-1}{k} \int_{0}^{R_{\textmd{W}} }\underbrace{\ln [\,r_{1}^{2} +b\,]}_{u} \underbrace{r_{1}^{2k+1} \textmd{d}r_{1} }_{dv}} \\ {{\kern 1pt} {\kern 1pt} {\kern 1pt} {\kern 1pt} {\kern 1pt} {\kern 1pt} {\kern 1pt} {\kern 1pt} {\kern 1pt} {\kern 1pt} {\kern 1pt} {\kern 1pt} {\kern 1pt} {\kern 1pt} {\kern 1pt} {\kern 1pt} {\kern 1pt} {\kern 1pt} {\kern 1pt} {\kern 1pt} {\kern 1pt} {\kern 1pt} {\kern 1pt} {\kern 1pt} {\kern 1pt} {\kern 1pt} {\kern 1pt} {\kern 1pt} {\kern 1pt} {\kern 1pt} {\kern 1pt} {\kern 1pt} {\kern 1pt} =2N\sum\limits_{k=0}^{N-1}\frac{(-1)^{k} }{R_{\textmd{W}}^{\,2k+2}} \binom{N-1}{k} \bigg(\underbrace{\frac{\ln [R_{\textmd{W}}^{\,2} +b\,] R_{\textmd{W}}^{\,2k+2} }{2k+2} }_{u{\kern 1pt} v} - \displaystyle\int_{0}^{R_{\textmd{W}} }\underbrace{\frac{2r_{1}^{2k+3} }{(2k+2)(r_{1}^{2} +b)} \textmd{d}r_{1}^{}}_{v{\kern 1pt} du}  \bigg)} \\ {{\kern 1pt} {\kern 1pt} {\kern 1pt} {\kern 1pt} {\kern 1pt} {\kern 1pt} {\kern 1pt} {\kern 1pt} {\kern 1pt} {\kern 1pt} {\kern 1pt} {\kern 1pt} {\kern 1pt} {\kern 1pt} {\kern 1pt} {\kern 1pt} {\kern 1pt} {\kern 1pt} {\kern 1pt} {\kern 1pt} {\kern 1pt} {\kern 1pt} {\kern 1pt} {\kern 1pt} {\kern 1pt} {\kern 1pt} {\kern 1pt} {\kern 1pt} {\kern 1pt} {\kern 1pt} {\kern 1pt} {\kern 1pt} = \underbrace{\ln [R_{\textmd{W}}^{\,2} +b\,]}_{2\ln R_{\textmd{W}} +\ln [\,1+\bar{b}\,]} -N\sum\limits_{k=0}^{N-1}\frac{(-1)^{k} }{R_{\textmd{W}}^{\,2k+2} (k+1)} \binom{N-1}{k} \bigg[\underbrace{(-1)^{k+1} b^{k+1} \ln [\,r_{1}^{2} +b\,]}_{(49-a)} +\underbrace{\frac{r_{1}^{2k+2} }{k+1} }_{(49-b)}} \\ {{\kern 1pt} {\kern 1pt} {\kern 1pt} {\kern 1pt} {\kern 1pt} {\kern 1pt} {\kern 1pt} {\kern 1pt} {\kern 1pt} {\kern 1pt} {\kern 1pt} {\kern 1pt} {\kern 1pt} {\kern 1pt} {\kern 1pt} {\kern 1pt} {\kern 1pt} {\kern 1pt} {\kern 1pt} {\kern 1pt} {\kern 1pt} {\kern 1pt} {\kern 1pt} {\kern 1pt} {\kern 1pt} {\kern 1pt} {\kern 1pt} {\kern 1pt} {\kern 1pt} {\kern 1pt} {\kern 1pt} {\kern 1pt} {\kern 1pt} {\kern 1pt} {\kern 1pt} {\kern 1pt} {\kern 1pt} {\kern 1pt} {\kern 1pt} {\kern 1pt} {\kern 1pt} {\kern 1pt} {\kern 1pt} {\kern 1pt} {\kern 1pt} {\kern 1pt} {\kern 1pt} {\kern 1pt} {\kern 1pt} {\kern 1pt} {\kern 1pt} {\kern 1pt} {\kern 1pt} {\kern 1pt} {\kern 1pt} {\kern 1pt} {\kern 1pt} {\kern 1pt} {\kern 1pt} {\kern 1pt} {\kern 1pt} {\kern 1pt} {\kern 1pt} {\kern 1pt} {\kern 1pt} {\kern 1pt} {\kern 1pt} {\kern 1pt} {\kern 1pt} {\kern 1pt} {\kern 1pt} {\kern 1pt} {\kern 1pt} {\kern 1pt} {\kern 1pt} {\kern 1pt} {\kern 1pt} {\kern 1pt} {\kern 1pt} {\kern 1pt} {\kern 1pt} {\kern 1pt} {\kern 1pt} {\kern 1pt} {\kern 1pt} {\kern 1pt} {\kern 1pt} {\kern 1pt} {\kern 1pt} {\kern 1pt} {\kern 1pt} {\kern 1pt} {\kern 1pt} {\kern 1pt} {\kern 1pt} {\kern 1pt} {\kern 1pt} {\kern 1pt} {\kern 1pt} {\kern 1pt} {\kern 1pt} {\kern 1pt} {\kern 1pt} {\kern 1pt} {\kern 1pt} {\kern 1pt} {\kern 1pt} {\kern 1pt} {\kern 1pt} {\kern 1pt} {\kern 1pt} {\kern 1pt} {\kern 1pt} {\kern 1pt} {\kern 1pt} {\kern 1pt} {\kern 1pt} {\kern 1pt} {\kern 1pt} {\kern 1pt} {\kern 1pt} {\kern 1pt} {\kern 1pt} {\kern 1pt} {\kern 1pt} {\kern 1pt} {\kern 1pt} {\kern 1pt} {\kern 1pt} {\kern 1pt} {\kern 1pt} {\kern 1pt} {\kern 1pt} {\kern 1pt} {\kern 1pt} {\kern 1pt} {\kern 1pt} {\kern 1pt} {\kern 1pt} {\kern 1pt} {\kern 1pt} {\kern 1pt} {\kern 1pt} {\kern 1pt} {\kern 1pt} {\kern 1pt} {\kern 1pt} {\kern 1pt} {\kern 1pt} {\kern 1pt} {\kern 1pt} {\kern 1pt} {\kern 1pt} {\kern 1pt} +\underbrace{\frac{r_{1}^{2k+2} }{2(k+1)} \ln (1+\bar{b})}_{(49-c)} +\underbrace{\sum _{k'=1}^{k}(-1)^{k+k'-1} \frac{r_{1}^{2k'} b^{k+1-k'} }{k'}}_{(49-d)} \bigg]_{ 0}^{\!R_{\textmd{W}}}} \end{array},
\end{equation}
with $\bar{b}=b/R_{\textmd{W}}^{2}$. The results associated with the terms $(49-a)$-$(49-d)$ are obtained as
\begin{equation} \label{eq_50_new}
\begin{array}{l} {(49-a)\to -N\ln [1+\frac{3}{2} e^{-\sigma _{z}^{2} } (N-2)] \sum\limits_{k=0}^{N-1}\frac{(-1)^{2k+1} }{R_{\textmd{W}}^{2k+2} (k+1)} \binom{N-1}{k} b^{k+1}} \\ {~~~~~~~~~~~~~ =\ln [1+\frac{3}{2} e^{-\sigma _{z}^{2} } (N-2)]\sum\limits_{k=0}^{N-1} \binom{N}{k+1} \bar{b}^{k+1}} \\ { ~~~~~~~~~~~~~ =\ln [1+\frac{3}{2} e^{-\sigma _{z}^{2} } (N-2)] \bigg[\sum\limits_{k=0}^{N} \binom{N}{k} \bar{b}^{k} -1 \bigg]} \\ { ~~~~~~~~~~~~~ =\ln [1+\frac{3}{2} e^{-\sigma _{z}^{2} } (N-2)] \left((1+\bar{b})^{N} -1\right) } \\ { ~~~~~~~~~~~~~ =\ln [1+\frac{3}{2} e^{-\sigma _{z}^{2} } (N-2)] \left((1+\frac{2e^{\sigma _{z}^{2} } }{3{\kern 1pt} (N-2)} )^{N} -1\right)} \end{array},
\end{equation}
\begin{equation} \label{eq_51_new}
(49-b) \to -N\sum\limits_{k=0}^{N-1} \frac{(-1)^{k} }{R_{\textmd{W}}^{2k+2}} \binom{N-1}{k} \frac{R_{\textmd{W}}^{2k+2} }{(k+1)^{2} } =-(\ln N+\gamma +\frac{1}{2N} ),
\end{equation}
\begin{equation} \label{eq_52_new}
(49-c)\to -N \sum\limits_{k=0}^{N-1}\frac{(-1)^{k} }{R_{\textmd{W}}^{2k+2} } \binom{N-1}{k} \frac{R_{\textmd{W}}^{2k+2} }{2(k+1)} \ln [1+\bar{b}]=-\frac{1}{2} \ln [1+\bar{b}],
\end{equation}
\begin{equation} \label{eq_53_new}
\begin{array}{l} {(49-d) \to -N\sum\limits_{k=0}^{N-1}\frac{(-1)^{k} }{R_{\textmd{W}}^{2k+2} (k+1)} \binom{N-1}{k} \sum\limits_{k'=1}^{k}(-1)^{k+k'-1} \frac{R_{\textmd{W}}^{2k} b^{k+1-k'} }{k'}} \\ { \approx -\big(\ln [N-1]+\frac{1}{2(N-1)} \big)\big((\frac{2}{3} e^{\sigma _{z}^{2} } )+\frac{1}{2!} (\frac{2}{3} e^{\sigma _{z}^{2} } )^{2} +\frac{1}{3{\kern 1pt} !} (\frac{2}{3} e^{\sigma _{z}^{2} } )^{3} + \cdots +\frac{1}{(N-1){\kern 1pt} {\kern 1pt} !} (\frac{2}{3} e^{\sigma _{z}^{2} } )^{N-1} \big){\kern 1pt} {\kern 1pt} {\kern 1pt} {\kern 1pt} } \\ { \mathop{\approx }\limits_{\sigma _{z}^{2} \le {\kern 1pt} {\kern 1pt} \ln (3/2)} -{\kern 1pt} {\kern 1pt} \big(\ln [N-1]+\frac{1}{2(N-1)} \big) \big(e^{(2e^{\sigma _{z}^{2} /3} )} -1 \big)} \end{array},
\end{equation}

The last line in~\eqref{eq_53_new} is obtained from the Taylor series approximation of $e^{x} \approx \sum _{{\kern 1pt} n=0}^{{\kern 1pt} N-1}x^{n} /n{\kern 1pt} ! $ for $x\le 1$. For further simplification, for large values of $N$, we approximate $e^{(2e^{\sigma _{z}^{2} } /3)} $ by $(1+2e^{\sigma _{z}^{2} } /(3(N-2)))^{N} =(1+\bar{b})^{N} $ using the identity $\mathop{\lim }\limits_{x\to \infty } {\kern 1pt} {\kern 1pt} (1+1/x)^{x} =e$. As a result,~\eqref{eq_53_new} can be approximated as $-\big(\ln [N-1]+1/(2(N-1))\big) \left((1+\bar{b})^{N} -1\right)$. Now, the equations in~\eqref{eq_53_new} and \eqref{eq_50_new} have a common term via $\left((1+\bar{b})^{N} -1\right)$ which can be factored out. Finally, after some manipulations, the integration denoted as $(**)$ in~\eqref{eq_48_new} is obtained as
\begin{equation} \label{eq_54_new}
\begin{split}
& \displaystyle\int_{0}^{R_{\textmd{W}}}\ln \big[1+\frac{b}{{\kern 1pt} r_{1}^{2} } \big] f_{R_{\,1} } (r_{1}) \,\textmd{d} r_{1} = \frac{1}{2} \ln [1+\bar{b}\,] \\
& ~~~~~~~~~~~~~~~~~~+\big((1+\bar{b})^{N} -1\big) \bigg(\ln \bigg[\frac{1+1.5e^{-\sigma_{z}^{2}} (N-2)}{N-1} \bigg] -\frac{1}{2(N-1)} \bigg)
\end{split},
\end{equation}
which along with \eqref{eq_47_new}, gives the worst ergodic user capacity expression in~\eqref{eq_35_new}.

\section{The Change in Performance for $d \ll R_{\textmd{W}}$}

In this Appendix we investigate how quickly the performance approaches the worst case as $d$ becomes small compared to $R_{\textmd{W}}$, i.e., $d \ll R_{\textmd{W}}$. In this case, an accurate approximation for the distance CDF $F_{R}(r)$ is given as
\begin{equation} \label{eq_36_update}
\begin{split}
F_{R\, | \, d}(r) \! \simeq \! \left\{
\begin{array}{l l c}
\!\! r^2 / R^{\,2}_{\textmd{W}} & ; & 0 \leq r \leq R_{\textmd{W}}-d \\
\!\! \bar{F}_{R}(R_{\textmd{W}}) \simeq 1-\frac{2d}{\pi R_{\textmd{W}}}  & ; & \! R_{\textmd{W}} - d  \leq r \leq R_{\textmd{W}} + d \\
\! \! 1 & ; & \!  R_{\textmd{W}} +d \leq r
\end{array}
\right.
\end{split}
\end{equation}
where $\bar{F}_{R}(\,\cdot\,)$ is given in~\eqref{eq_7}. The sub-script $|d$ in~\eqref{eq_36_update} means that the expression is evaluated at $d$. Now consider the distance CDF at the centre of network denoted as $F_{R\, | \, 0}(r)$. Let $\bar{d}$ denote a small value compared to $R_{\textmd{W}}$, i.e., $\bar{d} \ll R_{W}$. For a given $\bar{d}$, $F_{R\, | \, 0}(r)$ can be approximated as
\begin{equation} \label{eq_37_update}
\begin{split}
F_{R\, | \, 0}(r) \! \simeq \! \left\{
\begin{array}{l l c}
\!\! r^2 / R^{\,2}_{\textmd{W}} & ; & 0 \leq r \leq R_{\textmd{W}}-\bar{d} \\
\!\! 1-\frac{\bar{d}}{R_{\textmd{W}}}  & ; & \! R_{\textmd{W}} - \bar{d}  \leq r \leq R_{\textmd{W}} \\
\! \! 1 & ; & \!  R_{\textmd{W}} \leq r
\end{array}
\right.
\end{split}
\end{equation}
It follows that for any $\bar{d} = d \ll R_{\textmd{W}}$, we get $F_{R \,| \, d}(r) \simeq F_{R \, | \, 0}(r)\,, 0 \leq r \leq R_{\textmd{W}}$, and so $f_{R_{\,1}\, | \, d}(r_{\,1}) \simeq f_{R_{\,1}\, | \, 0}(r_{\,1}) \,, 0 \leq r_{\,1} \leq R_{\textmd{W}}$.

On the other hand, using the Taylor series expansion around $d = 0$, the two terms $(R_{\textmd{W}} + d)^{-\alpha}$ and $(R_{\textmd{W}} + d)^{-2\alpha}$ can be closely approximated as
\begin{equation} \label{eq_38_update}
(R_{\textmd{W}} + d)^{-\alpha} \simeq R^{\,-\alpha}_{\textmd{W}} - \alpha d R^{\,-\alpha-1}_{\textmd{W}} + \frac{\alpha(\alpha+1)}{2} d^{\,2} R^{\,-\alpha-2}_{\textmd{W}},
\end{equation}
\begin{equation} \label{eq_39_update}
(R_{\textmd{W}} + d)^{-2\alpha} \simeq R^{\,-2\alpha}_{\textmd{W}} - 2\alpha d R^{\,-2\alpha-1}_{\textmd{W}} + \alpha(2\alpha+1) d^{\,2} R^{\,-2\alpha-2}_{\textmd{W}} .
\end{equation}

From~\eqref{eq_38_update} -~\eqref{eq_39_update} and the approximation that $F_{R \,| \, d}(r) \simeq F_{R \, | \, 0}(r)\,, 0 \leq r \leq R_{\textmd{W}}$, the parameters $M_{\,1}$ and $ M_{\,2}$ (given in~\eqref{eq_21} -~\eqref{eq_22}) can be rewritten as
\begin{equation} \label{eq_40_update}
M_{\,1| \,d} \simeq M_{\,1| \,0} + \Delta M_{\,1 | \,d},
\end{equation}
\begin{equation} \label{eq_41_update}
M_{\,2| \,d} \simeq M_{\,2| \,0} + \Delta M_{\,2 | \,d},
\end{equation}
where $M_{\,1| \,0}$ and $M_{\,2| \,0}$ are the values of $M_{\,1}$ and $ M_{\,2}$ evaluated at $d = 0$, respectively. Also, $\Delta M_{\,1 | \,d}$ and $\Delta M_{\,2 | \,d}$ are the correction terms given as
\begin{equation} \label{eq_42_update}
\begin{array}{l} {
\Delta M_{\,1 | \,d} = (N-1)\sigma _{s}^{2} e^{\sigma _{z}^{2} /2} \bigg[ -\alpha d R^{\,-\alpha-1}_{\textmd{W}} + \frac{\alpha(\alpha+1)}{2} d^{\,2} R^{\,-\alpha-2}_{\textmd{W}} \! + \!  \displaystyle\int_{(R_{\textmd{W}} + d)^{-\alpha} }^{ R_{\textmd{W}}^{-\alpha}} \!\!\! G_{\,|\,d}(s_{j}^{-1/\alpha}) \,\textmd{d} s_{j} \bigg],}
\end{array}
\end{equation}
and
\begin{equation} \label{eq_43_update}
\begin{split}
& \Delta M_{\,2 | \,d} \! = \! 2(\!N-1)\sigma _{s}^{4} e^{2\sigma _{z}^{2}} \bigg[ \!\! - \! 2\alpha d R^{\,-2\alpha-1}_{\textmd{W}} \!\! + \!\! \alpha(2\alpha+1) d^{\,2} R^{\,-2\alpha-2}_{\textmd{W}} \!\! + \!\! \displaystyle\int_{(R_{\textmd{W}} + d)^{-2\alpha} }^{ R_{\textmd{W}}^{-2\alpha} } \!\!\!\! G_{\,|\,d}(s_{j}^{-1/2\alpha}) \,\textmd{d} s_{j} \! \bigg] \\
& + 4(N-1)(N-2)\sigma _{s}^{4} e^{\sigma _{z}^{2}} \bigg[ - \alpha d R^{\,-\alpha-1}_{\textmd{W}} \! + \! \frac{\alpha(\alpha+1)}{2} d^{\,2} R^{\,-\alpha-2}_{\textmd{W}} \! + \!  \displaystyle\int_{(R_{\textmd{W}} + d)^{-\alpha} }^{ R_{\textmd{W}}^{-\alpha}} \!\!\!\!\! G_{\,|\,d}(s_{j}^{-1/\alpha}) \,\textmd{d} s_{j} \bigg]^{2} \\
& +8(N-1)(N-2)\sigma _{s}^{4} e^{\sigma _{z}^{2}} \bigg[ R^{\,-\alpha}_{\textmd{W}} +  \int_{R_{\textmd{W}}^{-\alpha} }^{ r_{\,1}^{-\alpha}} G_{\,|\,d}(s_{j}^{-1/\alpha}) \,\textmd{d} s_{j} \bigg] \\
& ~~~~~~~~~~~~~~~~~~~~~~~~~~~ \times \bigg[ \! - \alpha d R^{\,-\alpha-1}_{\textmd{W}} \! + \! \frac{\alpha(\alpha+1)}{2} d^{\,2} R^{\,-\alpha-2}_{\textmd{W}} \! + \! \displaystyle\int_{(R_{\textmd{W}} + d)^{-\alpha} }^{ R_{\textmd{W}}^{-\alpha}} \!\!\!\! G_{\,|\,d}(s_{j}^{-1/\alpha}) \,\textmd{d} s_{j} \bigg],
\end{split}
\end{equation}
where $G_{\,|\,d}(\,\cdot \,) = (F_{R \, | \,d}(\,\cdot\,)-F_{R\,|\,d}(r_{1}))/(1-F_{R\,|\,d}(r_{1}))$.

Substituting~\eqref{eq_40_update} -~\eqref{eq_41_update} in~\eqref{eq_mu_I} -~\eqref{eq_sigma_I},
\begin{equation} \label{eq_44_update}
\begin{split}
& \mu _{\textmd{Denom} |\,d} \simeq 2[ \ln (M_{\,1| \,0} + \Delta M_{\,1 | \,d})] -0.5[\ln ( M_{\,2| \,0} + \Delta M_{\,2 | \,d} )] \\
& = \underbrace{2 \ln (M_{\,1| \,0}) -0.5 \ln ( M_{\,2| \,0})}_{\mu _{\textmd{Denom} \, |\,0}} + \underbrace{2 \ln (1+ \Delta M_{\,1 | \,d}/M_{\,1| \,0}) -0.5 \ln (1+ \Delta M_{\,2 | \,d} / M_{\,2| \,0})}_{\Delta \mu _{\textmd{Denom} \, |\,d}},
\end{split}
\end{equation}
\begin{equation} \label{eq_45_update}
\sigma _{\textmd{Denom} |\,d }^{2}  \simeq  \underbrace{-2 \ln (M_{\,1| \,0}) + \ln ( M_{\,2| \,0})}_{\sigma _{\textmd{Denom} \, |\,0 }^{2}} \underbrace{-2 \ln (1+ \Delta M_{\,1 | \,d}/M_{\,1| \,0}) + \ln (1+ \Delta M_{\,2 | \,d} / M_{\,2| \,0})}_{\Delta \sigma _{\textmd{Denom} \,|\,d }^{2}}.
\end{equation}

Since $\mu_{\textmd{Num}}$ and $\sigma_{\textmd{Num}}^{2}$ in~\eqref{eq_15} -~\eqref{eq_16} do not depend on $d$, for a given $r_{\,1}$, we get $\mathtt{SIR}_{\,r_{1}|\,d } \sim \LN(\mu _{\mathtt{SIR} \, | \,d} ,\sigma _{\mathtt{SIR} \, | \, d} )$ with
\begin{equation} \label{eq_46_update}
\mu _{\mathtt{SIR}\, |\, d} = \mu _{\textmd{Num}} - \mu _{\textmd{Denom} \,|\,d} \simeq \mu _{\mathtt{SIR}\, |\, 0} - \Delta \mu _{\textmd{Denom} \, |\,d},
\end{equation}
\begin{equation} \label{eq_47_update}
\sigma _{\mathtt{SIR} \, | \,d}^{2} = \sigma _{\textmd{Num}}^{2} + \sigma _{\textmd{Denom} \,|\,d}^{2} \simeq \sigma _{\mathtt{SIR} \, | \,0}^{2} + \Delta \sigma _{\textmd{Denom} \,|\,d }^{2}.
\end{equation}

As it is shown above, for $d \ll R_{\textmd{W}}$, we can accurately approximate each of the parameters $\mu _{\,\mathtt{SIR}\, |\, d}$ and $\sigma _{\mathtt{SIR} \, | \,d}^{2}$ by the corresponding value obtained at the centre of the finite-area network (with $d =0$) plus a correction term. However, due to the nonlinear structure of the conditional SIR coverage probability given as $\Px\{ \mathtt{SIR}_{\,r_{1} |\,d} > T \} = Q\left( (\ln T - \mu _{\mathtt{SIR}\,|\,d}) / \sigma _{\mathtt{SIR}\,|\,d} \right)$, it is not possible to have $\Px\{ \mathtt{SIR}_{\,r_{1} |\,d} > T \} \simeq \Px\{ \mathtt{SIR}_{\,r_{1} |\,0} > T \} + \Delta \Px\{ \mathtt{SIR}_{\,r_{1} |\,d} > T \}$. As an alternative, for the formulation tractability, we consider the mean of SIR (expressed in dB) averaged over different realizations of nearest AP locations, as the performance metric. It follows that
\begin{equation} \label{eq_48_update}
\begin{split}
\mathtt{SIR}^{\textmd{\,avg}}_{\,|\,d} (\textmd{dB}) & = \int_{0}^{R_{\textmd{W}} + d} 10 \log_{10}( \Ex \{ \mathtt{SIR}_{\,r_{1} \,|\,d} \} ) f_{R_{\,1\,|\,d}}(r_{1})\,d\,r_{1} \\
& =\frac{10}{\ln 10} \int_{0}^{R_{\textmd{W}} + d} (\mu _{\mathtt{SIR}\, |\, d} + \sigma _{\mathtt{SIR} \, | \,d}^{2}/2) f_{R_{\,1\,|\,d}}(r_{1})\,d\,r_{1}.
\end{split}
\end{equation}
Since $f_{R_{\,1}\, | \, d}(r_{\,1}) \simeq f_{R_{\,1}\, | \, 0}(r_{\,1}) \,, 0 \leq r_{\,1} \leq R_{\textmd{W}}$ and $f_{R_{\,1}\, | \, 0}(r_{\,1}) \simeq 0 \,, r_{\,1} \geq 0.75 R_{\textmd{W}}$ (see Fig.~\ref{Appendix_I_Fig_1} in Appendix A), Eq.~\eqref{eq_48_update} simplifies to
\begin{equation} \label{eq_49_update}
\begin{split}
\mathtt{SIR}^{\textmd{\,avg}}_{\,|\,d} (\textmd{dB}) \simeq & \underbrace{\frac{10}{\ln 10} \int_{0}^{R_{\textmd{W}}} (\mu _{\mathtt{SIR}\, |\, 0} + \sigma _{\mathtt{SIR} \, | \,0}^{2}/2) f_{R_{\,1\,|\,0}}(r_{1})\,d\,r_{1}}_{\mathtt{SIR}^{\textmd{\,avg}}_{\,|\,0} (\textmd{dB})} \\
& ~~~~+ \underbrace{\frac{10}{\ln 10} \int_{0}^{R_{\textmd{W}} + d} (- \Delta \mu _{\textmd{Denom} \, |\,d} + \Delta  \sigma _{\textmd{Denom} \,|\,d }^{2}/2) f_{R_{\,1\,|\,d}}(r_{1})\,d\,r_{1}}_{\Delta \mathtt{SIR}^{\textmd{\,avg}}_{\,|\,d} (\textmd{dB})}.
\end{split}
\end{equation}
where $\mathtt{SIR}^{\textmd{\,avg}}_{\,|\,0} (\textmd{dB})$ is the worst-case performance obtained at the centre of the finite-area network and $\Delta \mathtt{SIR}^{\textmd{\,avg}}_{\,|\,d} (\textmd{dB})$ is the associated correction term given as a function of $d$. Fig.~\ref{Delta_performance} illustrates $\Delta \mathtt{SIR}^{\textmd{\,avg}}_{\,|\,d} (\textmd{dB})$ (evaluated numerically) versus $d$ for different values of $\lambda$ in an interference-limited network with $\alpha =3.87$, $\sigma_L = 6~\textmd{dB}$, and $R_{\textmd{W}} = 1~\textmd{km}$.
\begin{figure}[t]
\begin{center}
\includegraphics[width = .5\textwidth]{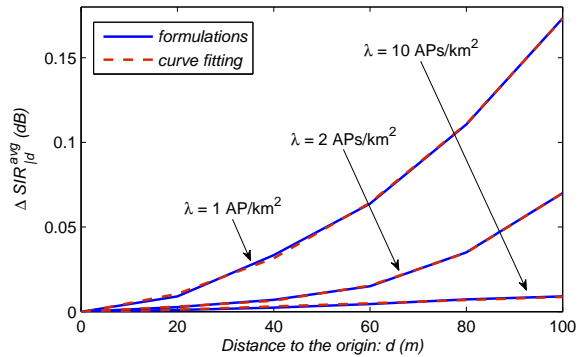}
\caption{Performance correction term $\Delta \mathtt{SIR}^{\textmd{\,avg}}_{\,|\,d} (\textmd{dB})$ versus $d$ for different values of $\lambda$ in an interference-limited network with $\alpha =3.87$, $\sigma_L = 6~\textmd{dB}$, and $R_{\textmd{W}} = 1~\textmd{km}$.}
\label{Delta_performance}
\end{center}
\end{figure}
In general, a smaller change in performance is observed for higher values of $\lambda$, to the extent that, in a dense interference-limited network there is only a very slight change in performance compared to the result obtained at the centre. Further, in order to quantify the change in performance, we may use a polynomial fit as a function of $d$ to provide a simple closed-form expression for $\Delta \mathtt{SIR}^{\textmd{\,avg}}_{\,|\,d} (\textmd{dB})$ as
\begin{equation} \label{eq_50_update}
\Delta \mathtt{SIR}^{\textmd{\,avg}}_{\,|\,d} (\textmd{dB}) \simeq  \sum_{i=0}^{n}a_{\,i}d^{\,n-i},
\end{equation}
where $a_{\,i}$, $i=1,\cdots,n$ are the coefficients found by curve fitting the numerical values obtained from formulations. The resultant curve-fitting approximations are also included in Fig.~\ref{Delta_performance} (the dotted lines) using a $3^{\textrm{rd}}$ order polynomial. The associated coefficients are given in Table~\ref{Tab1}. As is clear, for the example under consideration, a $3^{\textrm{rd}}$ order polynomial is adequate to describe $\Delta \mathtt{SIR}^{\textmd{\,avg}}_{\,|\,d} (\textmd{dB})$, though, if required, an even more accurate fit is possible with a higher polynomial degree.
\begin{table}
\caption{The coefficients, $a_{\,i}; i=0,\cdots,3$, of the polynomial in~\eqref{eq_50_update} for AP densities $\lambda=1~\textmd{AP/km}^2$, $\lambda=2~\textmd{APs/km}^2$ and $\lambda=10~\textmd{APs/km}^2$ in an interference-limited network with $\alpha =3.87$, $\sigma_L = 6~\textmd{dB}$, and $R_{\textmd{W}} = 1~\textmd{km}$.}
 \label{Tab1}
 \begin{center}
 \begin{tabular}{|c||c|c|c|c|}
   \hline
   Coefficients & $a_0$ & $a_1$ & $a_2$ & $a_3$\\
 \hline \hline
   $\lambda=1~\textmd{AP/km}^2$   & $3.91 \times 10^{-8}$ &        $1.13 \times 10^{-5}$ &        $3.31 \times 10^{-4}$ &       $-4.38 \times 10^{-4}$ \\  \hline
   $\lambda=2~\textmd{APs/km}^2$   & $1.06 \times 10^{-7} $ &        $ -5.94 \times 10^{-6} $ &          $ 2.31 \times 10^{-4} $ &      $ -7.77 \times 10^{-5} $ \\  \hline
   $\lambda=10~\textmd{APs/km}^2$  & $0$ &    $ 1.21 \times 10^{-8} $ &        $ 9.41 \times 10^{-5} $ &       $ -6.73 \times 10^{-4} $ \\  \hline
 \end{tabular}
 \end{center}
 \end{table}

\bibliography{ref}

\begin{thebibliography}{10}

\bibitem{mao2012towards}
G.~Mao and B.~Anderson, ``Towards a better understanding of large-scale network
  models,'' {\em IEEE/ACM Trans. Networking}, vol.~20, no.~2, pp.~408--421,
  2012.

\bibitem{wyner1994shannon}
A.~D. Wyner, ``Shannon-theoretic approach to a gaussian cellular
  multiple-access channel,'' {\em IEEE Trans. Inform. Theory}, vol.~40, no.~6,
  pp.~1713--1727, 1994.

\bibitem{goldsmith2005wireless}
A.~Goldsmith, {\em Wireless communications}.
\newblock Cambridge University Press, 2005.

\bibitem{andrews2011tractable}
J.~G. Andrews, F.~Baccelli, and R.~K. Ganti, ``A tractable approach to coverage
  and rate in cellular networks,'' {\em IEEE Trans. Comm.}, vol.~59, no.~11,
  pp.~3122--3134, 2011.

\bibitem{huang2013analytical}
K.~Huang and J.~G. Andrews, ``An analytical framework for multicell cooperation
  via stochastic geometry and large deviations,'' {\em IEEE Trans. Inform.
  Theory}, vol.~59, no.~4, pp.~2501--2516, 2013.

\bibitem{dhillon2012modeling}
H.~S. Dhillon, R.~K. Ganti, F.~Baccelli, and J.~Andrews, ``Modeling and
  analysis of k-tier downlink heterogeneous cellular networks,'' {\em IEEE J.
  Selected Areas in Comm.}, vol.~30, no.~3, pp.~550--560, 2012.

\bibitem{dhillon2013downlink}
H.~Dhillon, M.~Kountouris, and J.~Andrews, ``Downlink mimo hetnets: Modeling,
  ordering results and performance analysis,'' {\em IEEE Trans. Wireless
  Comm.}, vol.~12, no.~10, pp.~5208--5222, 2013.

\bibitem{di2013average}
M.~Di~Renzo, A.~Guidotti, and G.~Corazza, ``Average rate of downlink
  heterogeneous cellular networks over generalized fading channels: a
  stochastic geometry approach,'' {\em IEEE Trans. Comm.}, vol.~61, no.~7,
  pp.~3050--3071, 2013.

\bibitem{heath2013modeling}
R.~W. Heath, M.~Kountouris, and T.~Bai, ``Modeling heterogeneous network
  interference with using poisson point processes,'' {\em IEEE Trans. Signal
  Processing}, vol.~61, no.~16, pp.~4114--4126, 2013.

\bibitem{haenggi2009stochastic}
M.~Haenggi, J.~G. Andrews, F.~Baccelli, O.~Dousse, and M.~Franceschetti,
  ``Stochastic geometry and random graphs for the analysis and design of
  wireless networks,'' {\em IEEE J. Selected Areas in Comm.}, vol.~27, no.~7,
  pp.~1029--1046, 2009.

\bibitem{win2009mathematical}
M.~Z. Win, P.~C. Pinto, and L.~A. Shepp, ``A mathematical theory of network
  interference and its applications,'' {\em IEEE Proceedings}, vol.~97, no.~2,
  pp.~205--230, 2009.

\bibitem{chiu2013stochastic}
S.~N. Chiu, D.~Stoyan, W.~S. Kendall, and J.~Mecke, {\em Stochastic geometry
  and its applications}.
\newblock John Wiley \& Sons, 2013.

\bibitem{dhillon2013downlinkrate}
H.~Dhillon and J.~Andrews, ``Downlink rate distribution in heterogeneous
  cellular networks under generalized cell selection,'' {\em IEEE Wireless
  Comm. Letters}, vol.~3, no.~1, pp.~42--45, 2014.

\bibitem{madhusudhanan2012downlink}
P.~Madhusudhanan, J.~G. Restrepo, Y.~Liu, T.~X. Brown, and K.~R. Baker,
  ``Downlink performance analysis for a generalized shotgun cellular systems,''
  {\em IEEE Trans. Wireless Comm.}, 2014.

\bibitem{blaszczyszyn2013equivalence}
B.~Blaszczyszyn and H.~P. Keeler, ``Equivalence and comparison of heterogeneous
  cellular networks,'' in {\em IEEE Int. Symp. Personal, Indoor and Mobile
  Radio Comm.}, pp.~153--157, 2013.

\bibitem{vijayandran2012analysis}
L.~Vijayandran, P.~Dharmawansa, T.~Ekman, and C.~Tellambura, ``Analysis of
  aggregate interference and primary system performance in finite area
  cognitive radio networks,'' {\em IEEE Trans. Comm.}, vol.~60, no.~7,
  pp.~1811--1822, 2012.

\bibitem{srinivasa2007modeling}
S.~Srinivasa and M.~Haenggi, ``Modeling interference in finite uniformly random
  networks,'' in {\em Int. Workshop on Inform. Theory for Sensor Networks},
  2007.

\bibitem{salbaroli2009interference}
E.~Salbaroli and A.~Zanella, ``Interference analysis in a poisson field of
  nodes of finite area,'' {\em IEEE Trans. Vehicular Technology}, vol.~58,
  no.~4, pp.~1776--1783, 2009.

\bibitem{torrieri2012outage}
D.~Torrieri and M.~C. Valenti, ``The outage probability of a finite ad hoc
  network in nakagami fading,'' {\em IEEE Trans. Comm.}, vol.~60, no.~11,
  pp.~3509--3518, 2012.

\bibitem{torrieri2013analysis}
D.~Torrieri, M.~Valenti, and S.~Talarico, ``An analysis of the ds-cdma cellular
  uplink for arbitrary and constrained topologies,'' {\em IEEE Trans. Comm.},
  vol.~61, no.~8, pp.~3318--3326, 2013.

\bibitem{nitaigour2007sensor}
P.~Nitaigour, ``Sensor networks and configuration: Fundamentals, standards,
  platforms, and applications,'' 2007.

\bibitem{quek2013small}
T.~Q. Quek, G.~de~la Roche, and I.~G{\"u}ven{\c{c}}, {\em Small cell networks:
  Deployment, PHY techniques, and resource management}.
\newblock Cambridge University Press, 2013.

\bibitem{bazzi2006wlan}
A.~Bazzi, M.~Diolaiti, C.~Gambetti, and G.~Pasolini, ``Wlan call admission
  control strategies for voice traffic over integrated 3g/wlan networks,'' in
  {\em IEEE Consumer Comm. and Networking Conf.}, vol.~2, pp.~1234--1238, 2006.

\bibitem{liang2011dfmac}
H.~Liang and W.~Zhuang, ``Dfmac: Dtn-friendly medium access control for
  wireless local area networks supporting voice/data services,'' {\em Mobile
  Networks and Applications}, vol.~16, no.~5, pp.~531--543, 2011.

\bibitem{moller2004statistical}
J.~Moller and R.~P. Waagepetersen, {\em Statistical inference and simulation
  for spatial point processes}.
\newblock Chapman and Hall/CRC Press, 2004.

\bibitem{illian2008statistical}
J.~Illian, A.~Penttinen, H.~Stoyan, and D.~Stoyan, {\em Statistical analysis
  and modelling of spatial point patterns}, vol.~70.
\newblock John Wiley \& Sons, 2008.

\bibitem{rappaport1996wireless}
T.~S. Rappaport {\em et~al.}, {\em Wireless communications: principles and
  practice}, vol.~2.
\newblock Prentice Hall PTR New Jersey, 1996.

\bibitem{karakayali2006network}
M.~K. Karakayali, G.~J. Foschini, and R.~A. Valenzuela, ``Network coordination
  for spectrally efficient communications in cellular systems,'' {\em IEEE
  Wireless Comm.}, vol.~13, no.~4, pp.~56--61, 2006.

\bibitem{banani2013analyzing}
S.~A. Banani and R.~Adve, ``Analysing the reduced required bs density due to
  comp in cellular networks,'' in {\em IEEE Global Comm.}, pp.~2015--2019,
  2013.

\bibitem{pratesi2006generalized}
M.~Pratesi, F.~Santucci, and F.~Graziosi, ``Generalized moment matching for the
  linear combination of lognormal rvs: application to outage analysis in
  wireless systems,'' {\em IEEE Trans. Wireless Comm.}, vol.~5, no.~5,
  pp.~1122--1132, 2006.

\bibitem{singh2014joint}
S.~Singh, X.~Zhang, and J.~Andrews, ``Joint rate and sinr coverage analysis for
  decoupled uplink-downlink biased cell associations in hetnets,'' {\em IEEE
  Trans. Wireless Comm.}, 2015.

\bibitem{singh2013offloading}
S.~Singh, H.~S. Dhillon, and J.~G. Andrews, ``Offloading in heterogeneous
  networks: Modeling, analysis, and design insights,'' {\em IEEE Trans.
  Wireless Comm.}, vol.~12, no.~5, pp.~2484--2497, 2013.

\bibitem{haenggi2009interference}
M.~Haenggi and R.~K. Ganti, {\em Interference in large wireless networks}.
\newblock Now Publishers Inc, 2009.

\bibitem{khalid2012distance}
Z.~Khalid and S.~Durrani, ``Distance distributions in regular polygons,'' {\em
  IEEE Trans. Vehicular Technology}, vol.~62, no.~5, pp.~2363--2368, 2013.

\bibitem{bettstetter2004connectivity}
C.~Bettstetter, ``On the connectivity of ad hoc networks,'' {\em The computer
  journal}, vol.~47, no.~4, pp.~432--447, 2004.

\bibitem{boudreau2009interference}
G.~Boudreau, J.~Panicker, N.~Guo, R.~Chang, N.~Wang, and S.~Vrzic,
  ``Interference coordination and cancellation for 4g networks,'' {\em IEEE
  Comm. Magazine}, vol.~47, no.~4, pp.~74--81, 2009.

\bibitem{mehta2007approximating}
N.~B. Mehta, J.~Wu, A.~F. Molisch, and J.~Zhang, ``Approximating a sum of
  random variables with a lognormal,'' {\em IEEE Trans. Wireless Comm.},
  vol.~6, no.~7, pp.~2690--2699, 2007.

\bibitem{di2009further}
M.~Di~Renzo, F.~Graziosi, and F.~Santucci, ``Further results on the
  approximation of log-normal power sum via pearson type iv distribution: a
  general formula for log-moments computation,'' {\em IEEE Trans. Comm.},
  vol.~57, no.~4, pp.~893--898, 2009.

\bibitem{havil2003exploring}
J.~Havil and J.~Gamma, ``Exploring euler’s constant,'' 2003.

\end{thebibliography}
\end{document}